\documentclass[twocolumn,10pt]{IEEEtran}
\usepackage{bm,cite,algorithm,algorithmic,float,amsmath,amssymb}

\usepackage{graphicx, epstopdf}
\usepackage{graphicx}
\usepackage{cite}
\usepackage{balance}
\usepackage[utf8x]{inputenc}
\usepackage{fancyhdr}
\usepackage{subfigure}
\usepackage{dsfont}
\usepackage{calc}
\usepackage{xcolor}
\usepackage{subfigure}

\usepackage[justification=centering]{caption}
\usepackage{amsmath,amssymb,amsthm}
\usepackage{graphicx,epstopdf}
\usepackage{epsfig}	
\usepackage{amsfonts,balance}

\usepackage{lipsum}
\usepackage{marginnote}

\newtheorem{lemma}{Lemma}

\newtheorem{remark}{Remark}

\IEEEoverridecommandlockouts

\usepackage{graphicx,epstopdf}

\usepackage{epsfig}	
\usepackage{amsfonts,balance}
\usepackage{bbm}
\floatname{algorithm}{Algorithm}
\usepackage{lipsum}
\usepackage{url}
\newtheorem{theorem}{Theorem}

\newtheorem{corollary}{Corollary}

\begin{document}

	\title{RACH in Self-Powered NB-IoT Networks: Energy Availability and Performance Evaluation}
	
	%Spatio-Temporal Model and Random Access Analysis for Massive IoT Networks: A Stochastic Geometry Approach
	
	\author{\IEEEauthorblockN
		{ Yan Liu, ~\IEEEmembership{Graduate Student Member,~IEEE,}
			Yansha Deng,
			~\IEEEmembership{Member,~IEEE,}\\
			Maged Elkashlan,
			~\IEEEmembership{Senior Member,~IEEE,}
			Arumugam Nallanathan, ~\IEEEmembership{Fellow,~IEEE,}\\
			Jinhong Yuan, 
			~\IEEEmembership{Fellow,~IEEE,}
			and Ranjan K. Mallik,  ~\IEEEmembership{Fellow,~IEEE}
		}\\
		%{\small ${}^*$Department of Informatics, King's College London, London, UK\\
		%${}^{\ddagger}$ University of Electronic Science and Technology of China, Chengdu, Sichuan, China\\}
		\vspace{-0.2cm}

			\thanks{ Manuscript received February 21, 2020; revised July 1, 2020, September 14, 2020 and November 7, 2020; accepted November 22, 2020. 
			This work was supported in part by
			the U.K. Engineering and Physical Sciences Research Council (EPSRC) under Grant EP/R006466/1, and in part by the UK-India Education and Research Initiative (UKIERI) under Grant DST-UKIERI-2017-18-014.		
			This paper was
			presented in part at the IEEE International Conference on Communications (ICC), May, China, 2019 \cite{8761839}. 	
			The associate editor coordinating the review of this paper and approving it for publication was 
			Salman Durrani.		
			\textit{(Corresponding author: Yansha Deng.)}}

		\thanks{Y. Liu,  M. Elkashlan and A. Nallanathan are with School of Electronic Engineering and Computer Science, Queen Mary University of London, London E1 4NS, UK
			%N. Jiang, Y. Deng and A. Nallanathan are with Department of Informatics, King's College London, London, UK 
			(e-mail:\{yan.liu,  maged.elkashlan, a.nallanathan\}@qmul.ac.uk). }
		\thanks{Y. Deng is with Department  of  Engineering, King's College London, London WC2R 2LS, UK 
			 (e-mail:yansha.deng@kcl.ac.uk).}
		\thanks{J. Yuan is with School of Electrical Engineering and Telecommunications, University of New South Wales, Sydney NSW 2052, Australia
			(e-mail:j.yuan@unsw.edu.au).
			%(Corresponding author: Yansha Deng (e-mail:yansha.deng@kcl.ac.uk).)
		}
		\thanks{R. K. Mallik is with the Department of Electrical Engineering, Indian
			Institute of Technology Delhi, Hauz Khas,
			New Delhi 110016, India (e-mail: rkmallik@ee.iitd.ernet.in).}

	%	\vspace*{-0.2cm}
		
	}
	
	\maketitle
	
%	\vspace*{-2cm}
	\begin{abstract}
		NarrowBand-Internet of Things (NB-IoT) is a new 3GPP radio access technology designed to provide better coverage for a massive number of low-throughput low-cost devices in delay-tolerant applications with low power consumption.
		To provide reliable connections with extended coverage,  a repetition transmission scheme is introduced to NB-IoT during both Random Access CHannel (RACH) procedure and data transmission procedure. 
		To avoid the difficulty in replacing the battery for IoT devices, the energy harvesting is considered  as a promising solution to support energy sustainability in the NB-IoT network. 
		In this work, we analyze RACH success probability in a self-powered NB-IoT network taking into account the repeated preamble transmissions and collisions, where each  IoT device with data is active when its battery energy is sufficient to support the transmission.
		%the energy arrival at each IoT device as a Possion random process
		We model the temporal dynamics of the energy level as a birth-death process, derive the energy availability of each IoT device, and examine its dependence on the energy storage capacity and the repetition value. 
		We show that in certain scenarios, the energy availability remains unchanged despite randomness in the energy harvesting.
		We also derive the exact expression for the RACH success probability of a {randomly chosen} IoT device under the derived energy availability, which is validated under different repetition values via simulations.
		We show that the repetition scheme can efficiently improve the RACH success probability in a light traffic scenario, but only slightly improves that performance with very inefficient channel resource utilization in a heavy traffic scenario.
		%We  show that increasing the repetition value increases the RACH success probability, but decreases the repetition efficiency. 
		
		%and then study the energy available probabilities of different energy using strategies and repetition values. By doing so, we derive the exact expression for the RACH success probability of NB-IoT nerwork under time correlated interference and device energy availability, which is validated under different repetition values via practical packet evolution simulations. 
		%Using tools from stochastic process, we characterize the fraction of time each NB-IoT device kept in ON state as the energy availability.

	\end{abstract}
	
%	\vspace*{+0.3cm}
	
	\begin{IEEEkeywords}
		NB-IoT, RACH, collision, energy harvesting, stochastic geometry.
	\end{IEEEkeywords}
	
	% make the title area
	\maketitle
	
	%========================================================================
	\section{Introduction}
	The Internet of Things (IoT) is a novel paradigm that is rapidly gaining interest in modern wireless telecommunications to support connections of billions of miscellaneous innovative devices. 
	Third Generation Partnership Project (3GPP) has introduced several standards in its releases to improve support for Low Power Wide Area (LPWA) IoT connectivity \cite{ratasuk2015overview,kunz2015enhanced,diaz20163gpp}. 
	In Rel-13, EC-GSM-IoT (Extended Coverage-GSM-IoT)\cite{name2015} and LTE-MTC (LTE-Machine-Type-Communications)\cite{name2014} have been introduced to existing Global System for Mobile Communications (GSM)\cite{stuckmann2003gsm} and Long-Term Evolution (LTE)\cite{dahlman20134g} networks for better providing IoT devices,  respectively. 
	Another feature is {NarrowBand-Internet of Things (NB-IoT)}\cite{name2015} whose applications include smart metering, intelligent environment monitoring, logistics tracking, municipal light,  waste management,  and so on. 
	\subsection{NarrowBand-Internet of Things}
	NB-IoT is a new 3GPP radio-access technology developed from existing LTE functionalities, whereas some features of its specification deemed unnecessary for LPWA  IoT needs have been stripped out\cite{hasan2013random}. %,ericsson2016cellular}. 
	Because of this, NB-IoT can provide unique advantages for various IoT services over other technologies like 2G, 3G or LTE. LPWA networks mainly require deep/wide coverage, low power consumption, massive connections, and lower cost. The inherent characteristics of NB-IoT make it a good fit for LPWA deployment as shown in Fig. 1.
	\begin{figure}[htbp!]
		%\begin{center}
		\centering
		\includegraphics[width=3in]{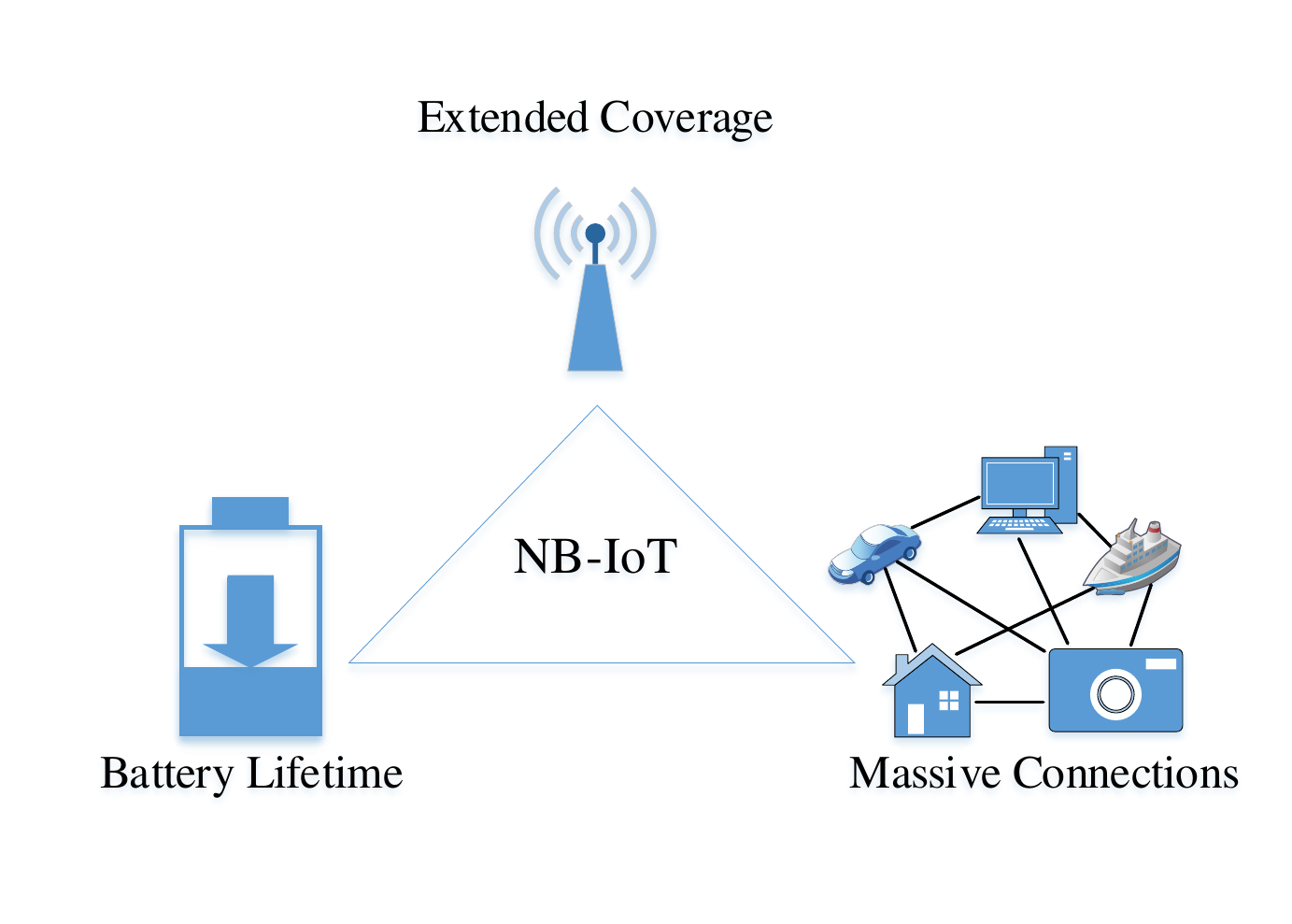}
		\caption{Main features of NB-IoT}
		\label{fig:my_label}
		%	\end{center}
	\end{figure}

	%Narrowband Internet of Things (NB-IoT) is an emerging cellular technology that will provide improved coverage for massive number of low-throughput low-cost devices with low device power consumption in delay-tolerant applications. 
	%NB-IoT is a new radio access system built from existing LTE functionalities with essential simplifications and optimizations.
	%and brieﬂy introduced in the following\cite{name2015}
	%Operating for many years on batteries is one of NB-IoT's biggest  challenges\cite{Summary2015}. 
	Extending battery lifetime is one of the features of NB-IoT.
	There are several ways to reduce power consumption and achieve lifespans of IoT devices  for more than 10 years. In the most simple way, we can switch the device to sleep mode when it does not work, so it doesn't waste power while waiting. Power management is fundamentally the tradeoff between message frequency, device sleep cycles, and business case needs.
	Deep coverage is another feature of NB-IoT, where it is designed to improve indoor coverage by 20 dB compared to conventional GSM/GPRS. This is achieved by a higher power density, as radio transmissions are concentrated on a narrower carrier bandwidth of just 180 kHz. The Coverage  Enhancement (CE) feature additionally offers the capability to repeat the transmission of a message when there exist poor coverage conditions but at the expense of a lower data rate. 
	%Additionally, NB-IoT is designed to support massive connections of devices.  RACH

	NB-IoT reuses the LTE design extensively, such that  the time required to develop full specifications and products is significantly reduced. 
	In the downlink of NB-IoT, OFDMA (Orthogonal Frequency Division Multiplexing) technology is adopted with the sub-carrier spacing of 15 kHz and 12 sub-carriers make up the 180 kHz channel\cite{schlienz2016narrowband}. In the uplink,  SC-FDMA (Single Carrier Frequency Division Multiple Access) technology including single sub-carrier and multiple sub-carrier are adopted. 
	A single sub-carrier technology with 3.75 kHz and 15kHz is adopted with carrier spanning over 48 and 12 respectively. 
	%Multiple subcarrier transmission supports sub-carrier spacing of 15 kHz and deﬁnes 12 continuous subcarriers which are combined into 3, 6, or 12.
	%The coverage ability for 3.75-kHz spacing is higher than for 15-kHz spacing because of higher power spectral density. 
	In this paper, we focus on NB-IoT Physical Random Access CHannel (NPRACH) with single sub-carrier spacing of 3.75 kHz. 
	
	NPRACH resources consist of the assignment of time and frequency resources and occur periodically. In the NPRACH, a random access preamble is transmitted which is the first step of random access procedure that enables the device to establish a data connection with its associated Evolved Node B (eNB)\cite{schlienz2016narrowband}. 
	To improve the quality of service and reduce the power  consumption of IoT devices, efficient RACH procedures need to be proposed and analyzed \cite{hasan2013random}. %,laya2014random}.
	In  \cite{luo1999stability,duan2013dynamic,R3}, mathematical models of contention-based RACH focusing on the Signal-to-Interference-plus-Noise Ratio (SINR) outage or collision have been studied.
	However, to the best of our knowledge, most works have focused either on studying the SINR outage without considering collision or studying the collision from the single cell point of view 
	and most results are based on the uplink power control due to its analysis simplicity.
	% either on studying the SINR outage from the network point of view without considering collision, or studying the collision prob- lem from the single cell point of view considering given fixed value of SINR outage
	% {
	% The authors in \cite{R3} combined queueing theory and stochastic geometry to analyze the stability region in a discrete-time slotted RA network. In \cite{R4}, the authors designed a user grouping and effective window size based
	%  RA protocol for the standalone LTE system in
	%  an unlicensed spectrum (SA LTE-U), where the UEs are divided into several
	%  groups, and at any time only one group is activated and allowed for its UEs to send RA attempt, which avoids the inter group UEs’ collision.}
	
	{
		Our previous work \cite{jiang2017random} has provided the preamble transmission model without considering the collision, \cite{nan2018collision} has considered the collision model without considering the repetition scheme, and \cite{8258982} has considered the repetition scheme based on the uplink power control for simplicity.
		However, the IoT device adopts the cell specific maximum transmission power without power control when transmitting more than two preamble repetitions,  which is motivated by the statement in 3GPP standard \cite{name2015}.
		In this scenario, it is unknown 1) to what extent the transmission power affects the RACH success; 2) to what extent the repetition transmission scheme affects the RACH success; 3) how to choose the repetition value in different traffic scenarios to balance the RACH success probability and data transmission channel resources. 
		To solve these problems, we present a novel mathematical framework to analyze and evaluate the RACH success probability, taking into account the SINR outage events as well as the collision events at the eNB, where the  IoT devices adopt fixed transmission power. }

	%{\color{red} expand the motivation of uplink power control based on NB-IoT standard}
	
	%NB-IoT is optimized for machine type traffic. It is kept as simple as possible in order to reduce device costs and to minimize battery consumption. In addition, it is also adapted to work in difficult radio conditions, which is a frequent operational area for certain machine type communication devices. Although NBIoT is an independent radio interface, it is tightly connected with LTE, which also shows up in its integration in the current LTE specifications.
	Generally speaking, in the NB-IoT network, the %physical layer characteristics 
	physical layer parameters and network topology
	can strongly affect the RACH performance, due to that the received SINR at the eNB can be severely degraded by the mutual interference generated from massive IoT devices. In this scenario, the random positions of the transmitters make accurate modeling and analysis of this interference even more complicated. 
	It is worth noting that stochastic geometry has been regarded as a powerful tool to model and analyze mutual interference between transceivers in the wireless networks\cite{elsawy2013stochastic}, 
	{such as conventional cellular networks \cite{novlan2013analytical}, wireless sensor networks \cite{deng2016physical}, cognitive radio networks \cite{6820767}, and heterogenous cellular networks \cite{6596082}.}
	
	{However, conventional stochastic geometry works \cite{elsawy2013stochastic,novlan2013analytical,deng2016physical,6820767,6596082}  focused on analyzing normal uplink or downlink scheduled data transmission channel, where the intra-cell interference is not considered, due to the ideal assumption that each orthogonal sub-channel is not reused in a cell. The work \cite{8587178} still focused on the scheduled data transmission, though it has considered the intra-cell interference over non-orthogonal sub-channels. 
		All these stochastic geometry works are different from the RACH analysis in this paper, where massive IoT devices in a cell may randomly choose and transmit the same preamble to their associated eNB to request for channel resources for data transmission, and different preambles represent orthogonal sub-channels. 
		Thus only IoT devices choosing the same preamble (i.e., sub-channel) have correlations. Note that IoT devices associated to the same eNB may choose the same preamble and collisions occur when two or more IoT devices transmit the same preamble simultaneously, such that the intra-cell interference and collisions are considered.}
	
	{In practical NB-IoT networks, mutual interference among transmissions is much more intricate than the conventional LTE network systems. In this paper, we develop a novel mathematical framework for NB-IoT networks using stochastic geometry for RACH analysis 
		with some challenges: 
		1) the distribution of the random transmission distances is considered, due to that the IoT device transmits with fixed transmission power;
		2) the intra-cell interference is considered, due to that the massive IoT devices in a cell may randomly choose and transmit the same preamble using the same sub-channel; 
		3) the temporally correlated mutual interference is considered, due to that each IoT device remains spatially static during each repetition;
		%This temporal correlation complicate the derivation of the preamble transmission success probability, which is the main challenge of RACH success analysis;
		4) both the SINR outage and the collision events are considered, due to that  collisions occur when two  or  more  IoT  devices  transmit  the  same  preamble  simultaneously;
		5) the network point of view analysis is considered instead of the single cell point of view, which is difficult in capturing both interference and collision generated from IoT devices transmitting with different transmit distances, due to the many concurrent transmissions and the interference experienced in each cell is different. 
		%conventional stochastic geometry works focused on analyzing normal uplink and downlink data transmission channel, where the intra-cell interference is not considered, due to the ideal assumption that each orthogonal sub-channel is not reused in a cell, whereas massive IoT devices in a cell may randomly choose and transmit the same preamble using the same sub-channel; 2) Without loss of generality, we assume each IoT device
	}

	%However, conventional stochastic geometry works\cite{} focused on analyzing normal uplink and downlink data transmission channel, where the intra-cell interference is not considered, due to the ideal assumption that each orthogonal sub-channel is not reused in a cell, whereas massive IoT devices in a cell may randomly choose and transmit the same preamble using the same sub-channel.
	%In practical NB-IoT networks, mutual interference among transmissions is much more intricate than the conventional LTE network systems. 
	%Without loss of generality, we assume each IoT device
	%remains spatially static during a TTI (i.e., we assume that the preamble format 0 is used, where a preamble repetition only takes 5.6 ms [1]), and thus the mutual interference among IoT devices are temporally correlated [7]. This temporal correlation complicate the derivation of the preamble transmission success probability, which is the main challenge of RACH success analysis
	
	%Therefore, we develop a spatial mathematical framework for NB-IoT networks using stochastic geometry, where the eNBs and IoT devices are modeled as independent PPPs.
	%{
	%It is challenging to derive the RACH success probability without power control because it is difficult to characterize the interference from interfering IoT devices depends on the different and random transmission distances, which has never been explored before.}

	\subsection{Energy Harvesting}
	Human-operated cellular devices, such as smart phones, can be charged at will. But IoT devices are often located at remote and hard to reach locations, such as underground or in tunnels, without access power supply, which may be inconvenient, dangerous, expensive, or even impossible to change the battery. 
	Hence, the battery energy storage highly determines the lifetime of the whole device, which provides a strong motivation for powering IoT devices by harvesting energy, such that networks consisting of energy harvesting or rechargeable batteries can survive perpetually\cite{jiang2005perpetual}.
	Practically, energy can be harvested from renewable environmental sources including thermal, solar, wind, etc.\cite{paradiso2005energy} and radio signals of different frequencies such as radio broadcasting \cite{tabassum_hossain_2016}.
	In these cases, the network performances are often tied closely to the efficiency in  energy harvesting and utilization.

	It is difficult to predict the time and the amount when the energy is available, as the energy arrival process is also random and dynamic. Energy buffer, i.e., battery storage, which collects harvested energy for signal processing and communication, is introduced in order to mitigate the unpredictability of the energy.
	Moreover, energy harvesting rates achievable today still fall short of typical power consumption levels. The harvested energy need to be accumulated in storage modules (e.g., capacitors or batteries) to a sufficient energy level to operate the IoT device. 
	Nevertheless, the energy buffer is limited. It remains challenging to assign the amount of energy in terms of uncertain energy sources. To have an in-depth study on the harvesting and utilization of energy, we need an analytically simple yet practically accurate model of the harvested energy. Several models on the harvested energy have been used in literature, assuming a time-slotted system.
	The arriving of harvested energy is known as a deterministic way in\cite{tacca2007cooperative} while unknown in\cite{kansal2007power}.
	In \cite{lei2009generic}\cite{ho2012optimal}, the stationary Markovian models of the harvested energy were studied, which are analytically simple and are thus useful to provide insights for solving some key theoretical problems. However, the validity of these models has not been formally justified with empirical measurements, and hence it is not known if the insights are useful in practice. In  \cite{ho2010markovian}, a more general analytical model for the harvested energy was provided with support from empirical measurements. In their empirical measurements, the harvested energy may be a Markovian process.
	In \cite{niyato2007sleep}, sleep/wake-up strategies for various factors were studied, which determines the optimal parameters of the solar energy harvest based strategy using a bargaining game model.
	However, to the best of our knowledge, there has been no work studying energy harvesting in NB-IoT networks. In this work, we consider NB-IoT networks with the IoT devices harvesting energy from nature, e.g., solar cells, microbial fuel cells, and water mills, etc.
	\subsection{Contributions and Outcomes}
	The contributions of this paper can be summarized as follows:
	
	1) Using stochastic geometry, we present a tractable analytical framework for the self-powered NB-IoT network via energy harvesting from natural resources, in which the locations of the eNBs and the IoT devices are modeled as two independent Poisson Point Processes (PPPs) in the spatial domain. 
	
	2) We model the arrival of harvesting energy as independent Poisson arrival processes. Using tools from Markov stochastic process, we characterize the fraction of time each IoT device kept in ON state as the energy availability. We first model the temporal dynamics of the energy level as a birth-death process, and then derive the expression of energy availability of each IoT device using hitting time analysis. 
	
	3) Based on the derived energy availability of the IoT device, we drive the expression for the RACH transmission success probability of {a randomly chosen} IoT device with fixed transmission power under both SINR outage and collision conditions in the NB-IoT network.
	Furthermore, we develop a realistic simulation framework to capture the randomness locations, the preamble transmission as well as the RACH collision, and verify our derived RACH success probability.
	
	4) Our results show that the energy availability of the IoT device increases at first and then remain unchanged, the RACH success probability increases, and the repetition efficiency decreases when increasing the repetition value.
	%Our results have shown  that the energy levels in each IoT device affect the energy availability, i.e., IoT devices need to harvest sufficient energy to a cutoff value (i.e., equals to the repetition value) before transmitting in the current time slot. 
	%For this setup, we further analyze this availability with different repetition values.
	{As such, the repetition  value needs to be optimized.} It is also noticed that in certain scenarios, the energy availability remains unchanged despite randomness in the energy harvesting.
	{The results also show that there is an upper limit on transmission power and too large transmission power will waste energy.}
	%%We show that there is a tradeoff between the repetition value and the cutoff value for reasonable energy storage capacities.  
	
	%4) Our numerical results have shown that increasing the repetition value increases the RACH success probability, but decreases the repetition efficiency.
	% Our numerical results shown that the RACH success probability increases with the repetition
	%{Our results have shown that increasing the repetition value increases the RACH success probability but decreases the repetition efficiency, and that higher repetition value will lead to the serious decrease of the  repetition efficiency in the light traffic scenario (fewer active IoT devices).}

	%{\color{red} What about the observations of the RACH transmission probability?}
	
	%5) In our work, we derived the exact expression for the RACH success probabilities of both a specific and a random IoT device.  As in many situations, merely considering the average performance sometimes reveals only limited information. And industry often focuses on the performance of the individual IoT devices. For example, in the IoT device side, the performances of different IoT devices could provide insights for parameters design for different IoT devices.
	
	The rest of the paper is organized as follows. Section II presents the network model. Section III derives energy availability of IoT device and analysis the actual result conditioning on some specific energy utilization strategies. Section IV derives the RACH transmission success probabilities of a {randomly chosen} IoT device. Section V presents the simulation framework. Finally, Section VI  concludes the paper.

	\section{System Model}
	%Note that stochastic geometry has been regarded as a powerful tool to model and analyze mutual interference between transceivers in the wireless networks, such as conventional cellular networks {\cite{andrews2011tractable,7037521,R3}}, wireless sensor networks \cite{deng2016physical}, cognitive radio networks \cite{deng2016artificial,elsawy2013stochastic}, and heterogenous cellular networks \cite{deng2016modeling,HElSawy2014Stochastic}.
	%{The authors in \cite{R3} using the stochastic geometry model to evaluate the stability region of large wireless networks. 
	%Similar as \cite{R3}, the number and locations of eNBs and IoT devices are fixed all time once they are deployed, thus the locations of active IoT devices are slightly correlated across time.
	% As the number and locations of BSs and IoT devices are fixed all time once they are deployed, 
	
	\subsection{Network Description}
	We consider an uplink stochastic geometry model for NB-IoT system, consisting of a single class of eNBs and IoT devices, that are spatially distributed in the Euclidean plane $\mathbb{R}^{2}$ following two independent homogeneous Poisson Point Processes (PPPs) $\Phi_{\rm{B}}$ and $\Phi_{\rm{D}}$ with intensities $\lambda_{\rm{B}}$ and $\lambda_{\rm{D}}$, respectively.
	{
		Without loss of generality, we assume each IoT device remains spatially static during time slots\cite{8258982}.} 
	%Similar as \cite{8258982}, the time is slotted into discrete time slots, and the number and locations of eNBs and IoT devices are fixed all time once they are deployed.
	Same as \cite{novlan2013analytical}, we assume that each IoT device associates with its geographically nearest eNB, where a Voronoi tesselation is formed. % and the eNBs are uniformly distributed in the Voronoi cell. 
	A standard power-law path-loss model is considered, where the path-loss attenuation is defined as $r^{-\alpha}$, with the propagation distance $r$ and the path-loss exponent ${\alpha}$. 
	In addition, we consider a Rayleigh fading channel, where the channel power gain $h$ is assumed to be exponentially distributed random variable with unit mean, i.e., $h\sim$ Exp(1). 
	All channel gains are assumed to be independent and identically distributed (i.i.d.) in space and time.
	% All the channel gains are independent of each other, independent of the spatial locations, and identically distributed (i.i.d.). %According to \cite{Tel2016,schlienz2016narrowband}, the transmitted power of IoT devices determined by the full path-loss inversion power contro l, where each IoT device compensates for its own path-loss to keep the average received signal power equal to a same threshold $\rho$. %By doing so, as a user moves closer to the desired base station, the transmission power required to maintain the same received signal power decreases, which saves energy for IoT devices. 
	
	%The transmission power of $i$th IoT device $P_i$ depends on the distance from its associated eNB, and the defined threshold $\rho$, where $P_i=\rho r_i^\alpha$. Without any loss of generality, in this paper we further assume that a $i$th IoT device transmits with a fixed power $P_i$ to simplify the notation and provide fundamental insights.
	\subsection{Random Access Procedure}
	In the uplink of NB-IoT, data can only be transmitted via the dedicated uplink data transmission channel, Narrowband Physical Uplink Shared CHannels (NPUSCH), which is scheduled by the associated eNB. Before resource scheduling, the IoT device needs to execute a RACH to request uplink channel resources with the associated eNB.
	The RACH procedure of NB-IoT has the simplified  message flow as for LTE, however, with different parameters\cite{schlienz2016narrowband}.
	There are two types of RACH procedures: contention-free and contention-based. The former is used to perform handover, whereas the latter is used otherwise %(e.g., RRC Connection Re-establishment)
	\cite{de2017random}. 
	%The contention-based random access procedure in NB-IoT consists of four steps:
	The four-step contention-based RACH used by NB-IoT is described as follows.
	%The UE first transmits a preamble (msg1) message on the random access channel during the first random access opportunity (RAO) after the triggering of the random access procedure. The eNB periodically informs the UEs about a set of up to 64 orthogonal preamble sequences from which the UE can make a choice. Collisions occur when two or more UEs transmit the same preamble sequence during the same RAO. However, the eNB does not detect such collisions during this step. 

	In step 1, the device first transmits a preamble (msg1) on the NPRACH during the first Random Access Opportunity (RAO). The eNB periodically informs the devices about a set of up to 64 orthogonal preamble sequences from which the IoT device can make a choice. Collisions occur when two or more devices transmit the same preamble sequence simultaneously. 
	In step 2, the IoT device sets a Random Access Response (RAR) window and waits for the eNB to transmit a RAR (msg2) with an uplink grant for the transmission of a message in the following step. If the device that sent a preamble sequence does not receive a msg2 from the eNB within a certain period of time, it enters a backoff period, trying to access the network once this period has expired.
	In step 3, the IoT device that successfully receives its RAR transmits a Radio Resource Control (RRC) Connection Request (msg3) with identity information to eNB. If two or more devices have chosen the same preamble sequence in a RAO, they will receive the same grant in the RAR message, and thus, their msg3 transmissions will collide. 
	In step 4, the eNB transmits a RRC Connection Setup (msg4) to the IoT device when it successfully receives msg3.
	More details on the RACH can be found in \cite{dahlman20134g}\cite{jiang2017random}.
	%\cite{zheng2012radio}
	
	%In the uplink(UL) of NB-IoT, the Narrowband Physical Uplink Shared CHannels (NPUSCHs) are used for data transmission and the Narrowband Physical Random Access CHannels (NPRACHs) are used for preamble transmission\cite{schlienz2016narrowband}. 
	Massive connections in NB-IoT make the simultaneous RACH requests under a limited number of available preambles one of the main challenges, thus we focus on the contention of preamble in step 1 of contention-based RACH, with the assumption that steps 2, 3, and 4 of RACH are always successful whenever step 1 is successful. If step 1 in RACH fails, i.e. the associated RAR message was not received, the IoT device needs to transmit another preamble in the next available RAO. That is to say, a RACH procedure is always successful if the IoT device successfully transmits the preamble to its associated eNB. In this case, the failure of this preamble transmission can result from the following two reasons: 1) the eNB cannot decode the preamble due to the low received SINR; 2) the eNB successfully decoded the same preamble from two or more IoT devices in the same time and causes the collision.
	
	It is known that the collision event in step 1 of RACH can be detected by the eNB, when the collided IoT devices are separable in terms of the power delay profile \cite{dahlman20134g}. Our model follows the assumption of collision handling in \cite{nan2018collision}, where collision events are detected by  eNB after it decodes the preambles in step 1 of RACH, and then no response will be fed back from the  eNB to the IoT devices, such that it can not proceed to the next step of RACH \cite{lin2016estimation}.
	%We assume that a RACH procedure is always successful if the IoT device successfully transmits the preamble to its associated eNB. 
	%%Without loss of generality, we assume that each IoT device associated to the same eNB has the same number of different preambles to choose from. Let $N_p$ denotes the number of available preambles. Each IoT device has an equal probability $1/N_p$ to choose a specific preamble, and then the average density of the IoT devices using the same preamble is ${\lambda _{\rm{D}p}} = {\lambda _{\rm{D}}}/{N_p}$ unit devices/preamble/km$^2$.
	
	\subsection{Narrowband Physical Random Access CHannel}
	NB-IoT technology occupies a frequency band of 180 kHz bandwidth\cite{name2016}, which corresponds to one resource block in LTE transmission. % For the UL, the NPUSCHs are used for data transmission and the Narrowband-PRACHs (NPRACHs) are used for preamble transmission. 
	% In the UL of NB-IoT, SC-FDMA either with a 3.75 kHz tone spacing over 48 subcarriers or 15 kHz tone spacing over 12 subcarriers is applied, where the NPRACH only supports 3.75 kHz tone spacing. 
	In the NPRACH, a preamble is transmitted based on symbol groups on a single subcarrier. 
	% Each symbol group has a cyclic prefix (CP) followed by 5 symbols. Two preamble formats, 0 and 1 are defined in NB-IoT, which varies in their CP length. 
	% The five symbols have a duration of $T_{\rm{SEQ}}= 1.333$ ms, prepended with a CP of $T_{\rm{CP}} = 67$ $\mu$s for format 0 and $267$ $\mu$s for format 1, giving a total length of $1.4$ ms and $1.6$ ms, respectively.
	%The NPRACH is a set of consecutive subcarriers which is periodically allocated speciﬁcally for preamble transmission. 
	A preamble consists of four symbol groups transmitted without gaps and can be repeated several times using the same transmission power. Frequency hopping is applied to symbol group granularity, i.e. each symbol group is transmitted on a different subcarrier.
	To serve UEs in different coverage classes, the NB-IoT network can configure up to 3 NPRACH resource configurations in a cell. In each configuration, a repetition value from the set $\{$1, 2, 4, 8, 16, 32, 64, 128$\}$ is specified for repeating a basic  preamble\cite{wang2017primer}.
	In this model, we consider a single repetition value $N_{\rm{T}}$ as \cite{8258982}, where the channel resources assignment of NPRACHs only takes place at the beginning of each transmission time interval (TTI, i.e., a time slot) as shown in \cite{8258982}.

	\subsection{Energy Harvesting Model}
	We assume that each IoT device is supported solely by the energy harvested from the surrounding environment (e.g., solar, kinetic, wind) and is equipped with a rechargeable battery with the finite capacity to buffer energy. Note that the assumption of a finite energy buffer is realistic since it is not possible to have infinite energy within an IoT device with limited physical dimensions. 
	%Considering a time scale scenario that is much longer than the one over which cell selection decision is taken, 
	%We model the energy arrival process of the the $i$th NB-IoT device as an independent Poisson process with mean $\mu_i$ units energy per unit time.which has been validated in literature using empirical measurements for a variety of real world energy harvesting module \cite{roundy2003energy}.
	We model the energy arrival process of a {randomly chosen} IoT device in one TTI as an independent Poisson process with intensity $\mu_0$.
	This assumption is based on the fact that most energy harvesting modules contain small sub-modules harvesting energy independently,  e.g.,  small solar cells harvesting energy in a  solar panel,  where the net energy harvested can be argued to be a Binomial process, which approaches to the Poisson process in the limit when the number of sub-modules grows large.
	This assumption is not uncommon, e.g., see \cite{5992841}\cite{8437573}. 
	% More general arrival processes can be found in \cite{5779157} for future work.
	
	%While the absolute units of energy are irrelevant, we assume that the normalization is such that each repetition requires one unit of energy. 
	%While the absolute units of energy are irrelevant, we assume that whenever a repetition is transmitted with power $P_i$ in an epoch of duration $T$, it will cost of $TP_i$ units of energy depletion.The bandwidth is sufﬁciently wide so that T can take small values and we approximate the slotted system to a continuous time system. 
	%This assumption can be easily relaxed to in corporate repetitions requiring more than one unit of energy under sufﬁcient randomization, but this case is not in the scope of the current work. 
	%The energy arrival process at the $i$th IoT device is modeled as a Poisson process\cite{roundy2003energy} with mean $\mu_i$  and is independent of that of other devices.
	%Incorporating more general arrival processes, e.g., the ones proposed in [26] for solar-powered nodes, is left for future work. 
	Since the energy arrivals are random and the energy storage capacities are finite, there is some uncertainty associated with whether the IoT device has enough energy to serve itself at a particular time. 
	Under such a constraint, the IoT device needs to be kept OFF, and be allowed to recharge before it has sufficient energy to serve itself in a given time slot. %Besides, as discussed in the sequel, it may also be preferable to keep a eNB OFF despite having enough energy. 
	Thus, at any given time, an IoT device can be in either of the two operational states: \rm{ON} or \rm{OFF}.  
	In this paper, the decision to toggle the operational state of an IoT device, i.e., turn \rm{ON} or \rm{OFF}, is taken by the device independently, and it is not influenced by the operational states of other IoT devices. For example, one IoT device may decide to turn \rm{OFF} if its current energy level reaches below a certain predefined level, and to turn \rm{ON} after harvesting enough energy over the other threshold. 
	%The IoT device may additionally consider the time for which it is in the current state while making the decision. For instance, a IoT device may start a timer whenever the state is toggled and may decide to toggle it back when the timer expires or the energy level reaches a certain minimum value, whichever occurs first. 
	%In this work, a IoT device is said to be available if it is in the ON state and has enough energy for the transmission of at least $N_T$ times of repetitions. i.e., has $E_{ic}$ units of energy, %\cite{yu2015energy}, where $E_{ic}$ is the predefined cutoff value.   
	
	We define the fraction of time that  a {randomly chosen} IoT device remains \rm{ON} as the $energy$ $availability$ $\eta_0$. In order to obtain $\eta_0$, we first need to characterize how the energy available at the IoT device changes over time.
	We model the available energy of an IoT device as a finite-state continuous-time Markov process (CTMC). In particular, let $\{M(t):t\geqslant 0\}$ be a stationary, homogeneous, and irreducible Markov process with state space $M = \{m_1,m_2,\cdots,m_{{i}}\}$ that specifies the energy state at time $t$.
	For this setup, we define the energy required of a {randomly chosen} IoT device in $each$ $repetition$ $E_{0}$ as the $unit$ $energy$\cite{yu2015energy}.
	% $E_{i0}$ = $E_{i0}^{RA}$ + $E_{i0}^{DA}$ for successful RACH as shown in Fig .2 (b)and $E_{i0}^{RA}$ is the energy required for the RACH in $each$ $repetition$; $E_{i0}^{DA}$ is the energy required for the data transmission in $each$ $repetition$.
	%For this setup, we define the energy required for the RACH and data transmission  of the $i$th IoT device in $each$ $repetition$ as unit energy $E_{i0}$. 
	% {According to \cite{yu2015energy}, $E_{i0}=PT$, where $P$ is the transmission power of each IoT device and $T$ = $T_{r}$ + $T_{g}$ is the duration of RACH and data transmission in each $each$ $repetition$ as shown in Fig. 2(b).}
	% \begin{figure}
	% 	\centering
	% 	\includegraphics[width=3.5in,height=2in]{energyCapture.PNG}
	% 	\caption{Structure of NPRACH %($\lambda_{\rm{D}}=10$ devices/km$^2$, $\lambda_{\rm{B}}=0.05$ eNBs/km$^2$, $\gamma_{th}= 0$ dB, $\alpha = 4$, and $P=20$ mW)
	% 	}
	% 	\label{fig:9}
	% \end{figure}
	% Due to the main focus of this paper is analyzing the contention-based RACH in NB-IoT network, we assume that the actual intended packet transmission is always successful if the corresponding RACH succeeds.
	%that the energy required for data transmission is a constant and 
	Thus,  the real energy state at a {randomly chosen} IoT device is directly proportional to the $E_{0}$ and the proportionality implies the maximum repetition value the IoT device can support.
	Without loss of generality, discretizing the real energy state of the IoT device by dividing by $E_{0}$, we get our state space of a {randomly chosen} IoT device as $M= \{0,1,\cdots,M_0\}$, in which
	$M_0=\bigl\lfloor{E}/{E_{0}}\bigr\rfloor$ and $E$ is the real energy storage capacity of the {randomly chosen} IoT device.
	Thus, the energy state $m\in M$ of the IoT device implies the maximum repetition value the IoT device can support with the energy.
	As such, the temporal dynamics of the battery energy levels can be modeled by the CTMC, in particular, the birth-death process illustrated in Fig. 2. When the IoT device is  \rm{ON} , the energy increases according to the energy harvesting rate $\mu_0$ units energy per second and decreases at a depletion rate of $\nu_0$ units energy per second.
	% where $\nu_i$=$\mathcal{A}_aP/E_{i0}$ and $\mathcal{A}_a$ is the non-empty probability (i.e., IoT device data buffer is non-empty) described in Section IV.
	
	\begin{figure}[htbp!]
		\centering
		\includegraphics[width=3.5in]{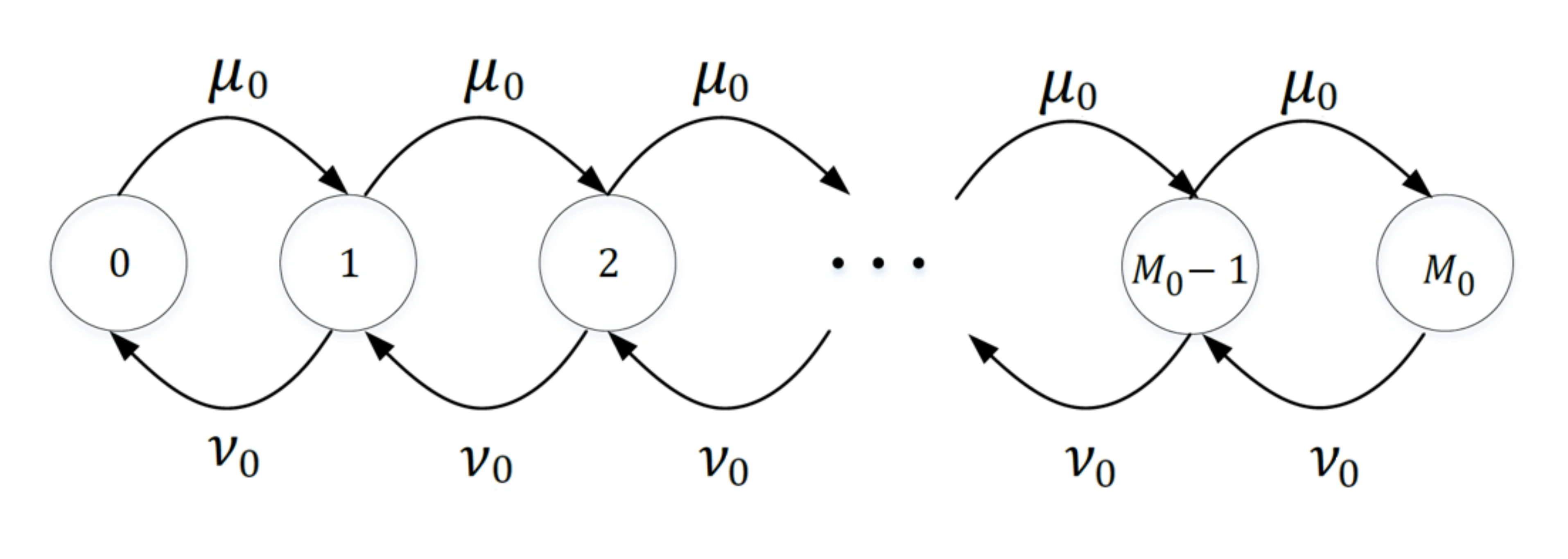}
		\caption{Birth-death process modeling the dynamics of the energy available at an IoT device}
		\label{fig:my_label1}
	\end{figure}

	Note that the data transmission after a successful RACH can be extended following the analysis of RACH success probability. 
	Since the main focus of this paper is analyzing the contention-based RACH in the NB-IoT network, we assume that the actual intended packet transmission is always successful (i.e., the data transmission success probability is one) if the corresponding RACH succeeds and that the RACH either always fails as shown in Fig. 3(a) or succeeds as shown in Fig. 3(b) to obtain the lower bound and upper bound{\footnote{{In our model, we take into account the RACH in one TTI with a single packet sequence transmission as shown rather than multiple TTI transmissions. We need to know the energy availability in the first TTI.
	To calculate energy availability for simplicity, we assume that the RACH either always fails as shown in Fig .3(a) or succeeds as shown in Fig. 3(b) to obtain the lower bound and upper bound of the $energy$ $availability$ $\eta_0$.
	The results of our work are the foundation of further  multiple TTIs analysis.}}} of the $energy$ $availability$ $\eta_0$.
	Thus, we have $E_{0}^f$ = $E_{0}^{RA}$ for failure RACH and $E_{0}^s$ = $E_{0}^{RA}$ + $E_{0}^{DA}$ for successful RACH, where $E_{0}^{RA}$ is the energy required for the RACH in $each$ $repetition$; $E_{0}^{DA}$ is the energy required for the data transmission in $each$ $repetition$.
		\begin{figure}[htbp!]
		\centering
		\subfigure[]
		{\begin{minipage}[t]{0.48\textwidth}
				\includegraphics[width=2.2in,height=1.8in]{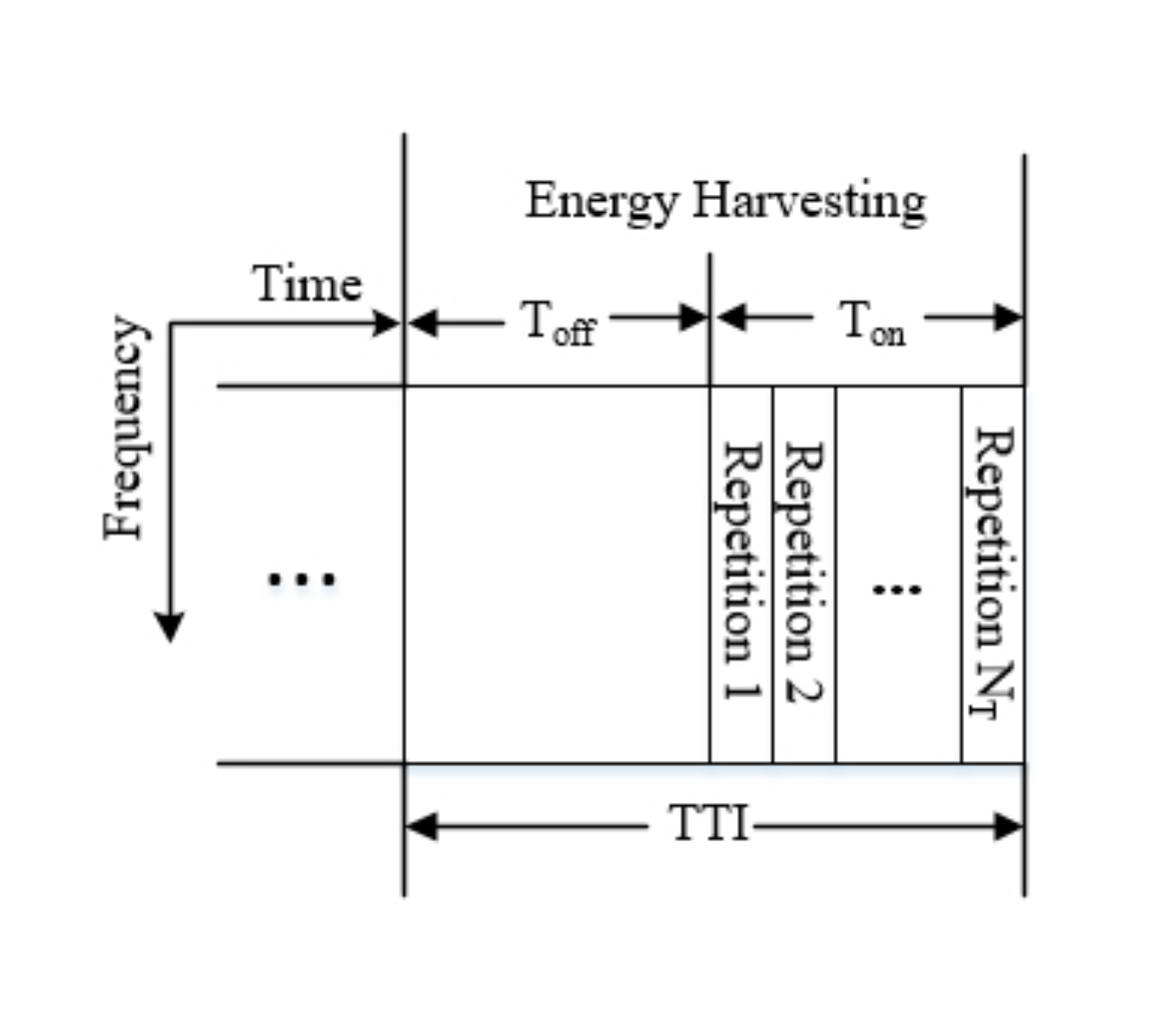}
		\end{minipage}}
		\subfigure[]
		{\begin{minipage}[t]{0.48\textwidth}
				\centering
				\includegraphics[width=3.5in,height=1.8in]{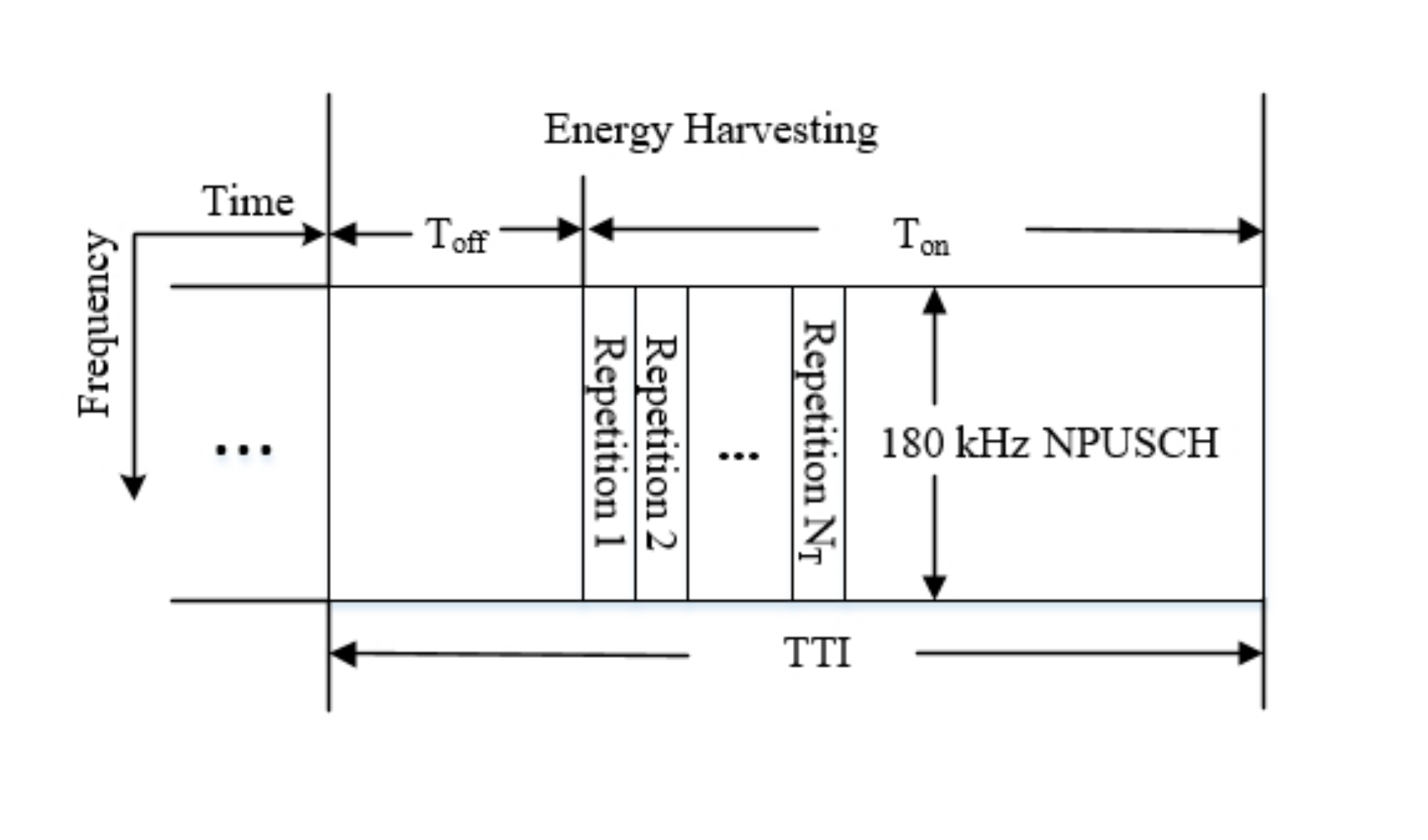}
		\end{minipage}}
		\caption{(a) Structure of failure NPRACH; (b) Structure of successful NPRACH and NPUSCH} 
		\label{fig:2}
	\end{figure}
	Thus, we have incorporated the energy depletion rate for failure RACH and successful RACH cases separately as follows:
	\begin{align}\label{nu_i}
	\begin{cases}
	{\nu_0} = {\mathcal{A}_a}P/E_{0}^{f},  &\mbox{if RACH fails},\\
	{\nu _0} = {\mathcal{A}_a}P/E_{0}^{s}, &\mbox{if RACH succeeds},
	\end{cases}
	\end{align}
	where $P$ is the transmission power of each IoT device and $\mathcal{A}_a$ is the non-empty probability (i.e., IoT device data buffer is non-empty) described in Section IV.
	
	The main notations of the proposed protocol are summarized in TABLE I.
	
	%( i.e., we assume to use preamble format 0 with $T=5.6$ ms).

	\begin{table}[htbp!]
		\centering
		\caption{Notation Table}
		{\renewcommand{\arraystretch}{1.2}
			%\renewcommand{\tabcolsep}{0.15cm}
			%	\begin{tabular}{|c|c|c|c|}
			{
				\begin{tabular}{|p{1cm}|p{7cm}|}
					\hline
					$\lambda_B$ & The intensity of BSs  \\ \hline $\lambda_D$ & The intensity of IoT devices  \\ \hline
					$h$ & The Rayleigh fading channel power gain  \\ \hline
					$r$ & The distance between an IoT device and its associated BS \\ 
					\hline
					$\alpha$ & The path-loss exponent    \\ \hline
					$N_T$ & The RACH repetition value     \\ \hline
					$\mu_0$  &  The energy arrival rate  \\ \hline $E_0$ & The energy required for each repetition\\ 
					\hline
					$E$ &  The real energy storage capacity  \\ \hline 
					$\nu_0$ &  The energy depletion rate  \\ 
					\hline
					$E_{0}^{RA}$ &  The energy required for the RACH in $each$ $repetition$  \\ \hline
					$E_{0}^{DA}$  &  The energy required for the data transmission in $each$ $repetition$ \\ 
					\hline
					$\mathcal{A}_a$ &  The non-empty probability \\ \hline 
					$P$ & The transmission power of each IoT device \\ 
					\hline
					$M_{0}^{c}$ & The cutoff value  \\ \hline
					$M_0$ & The storage  capacity  \\ 
					\hline
					$T_{0}^{\rm{off}}$ & The time for which the IoT device remains in the \rm{OFF} state  \\ \hline 
					$T_{0}^{\rm{on}}$ & The time for which the IoT device remains in the \rm{ON} state\\
					\hline
					$\sigma ^2$ & The noise power  \\ \hline  
					$\mathcal {I}_0$ & The aggregate  interference of the typical IoT device\\
					\hline	
					$\mathcal Z_{D}$ & The set of interfering IoT devices  \\ \hline 
					$\gamma _{th}$ & The SINR threshold  \\
					\hline	
					$L$ & The number of available preambles are reserved for the contention-based RACH  \\ \hline
					${\lambda _{Da}} $ & The density of active IoT devices choosing the same preamble\\
					\hline 
					$c$ & $c=3.575$ is a constant  
					\\ \hline  $\eta_0$ & The  energy  availability\\
					\hline
				\end{tabular}
			}
		}
		\label{table_accord}
	\end{table}

	\section{Energy Availability of IoT Devices}
	We assume that the energy harvesting processes are independent among the IoT devices, which ensure the independence of the current operational state (\rm{ON} or \rm{OFF}) of the IoT devices. An IoT device toggles its operational state solely on its current energy level.
	%\subsection{Availability for General Operational Strategy}
	Essentially, we focus on a general strategy ${S\{M_{0}^{min},M_{0}^{c}\}}$ with energy storage capacity $M_{0}$ ($0\le M_{0}^{min}\le M_{0}^{c}\le M_0$), where the {randomly chosen} IoT device toggles to \rm{OFF} state when its energy level reaches below specific level $M_{0}^{min}$ and toggles back to \rm{ON} state when its energy level reaches the predefined cutoff value $M_{0}^{c}$, i.e., it has sufficient energy to use. In addition, it should be noted that the cutoff value $M_{0}^{c}$ and $M_{0}^{min}$ can be defined and changed by the network if necessary. 
	
	We note that, it is strictly sub-optimal when the IoT device toggles to \rm{OFF} at $M_{0}^{min}\ne 0$ for this model, since it effectively reduces the storage capacity from $M_0$ to $M_0-M_{0}^{min}$. That is to say, the strategy ${S_1\{M_{0}^{min},M_{0}^{c}\}}$ with energy storage capacity $M_0$ is equivalent to 
	${S_2\{0,M_{0}^c-M_{0}^{min}\}}$ with energy storage capacity ($M_0-M_{0}^{min}$), which are described in Table II. 
	\begin{table}[htbp!]
		\centering
		\caption{ Strategy Summary}
		{\renewcommand{\arraystretch}{0.9}
			\begin{tabular}{|c|c|c|c|}
				\hline
				%Strategy &   $M_i^{ * \min } $ & $M_i^{ * c}$  %  & $M_i^ * $ \\ \hline
				${S_1}(M_0^{\min },M_0^c)$ & $M_0^{\min }$ & $M_0^{c}$ & $M_0 $\\ \hline
				${S_2}(0,M_0^c - M_0^{\min })$  &  0 & $M_0^c - M_0^{\min }$ & $M_0^{} - M_0^{\min }$  \\ \hline
				${S_3}(M_0^c - {N_T},M_0^c)$  & $M_0^c - {N_T}$ & $ M_0^c$ &  $M_0$ \\ \hline
				${S_4}(0,{N_T})$ & 0 & $N_T$  & {${M_0}$}\\ \hline
			\end{tabular}
		}
		\label{table_accord}
	\end{table}
	In our work, the IoT device will be allowed to transmit only if it harvests sufficient energy for the transmission of at least $N_T$ times of repetitions. Then we set $M_{0}^c\geqslant N_T$ and $M_{0}^{min}=M_{0}^c-N_T\geqslant0$, so we have our strategies as ${S_3\{M_{0}^c-N_T,M_{0}^c\}}$ with energy storage capacity $M_0$.
	Therefore, without loss of generality, the strategies ${S_3\{M_{0}^c-N_T,M_{0}^c\}}$ with energy storage capacity $M_0$ for our work is equivalent to ${S_4\{0,N_T\}}$ with energy storage capacity {$M_0$}. For the brevity of exposition, we denote this strategy by ${S\{N_T\}}$.
	In practical, the probability that the IoT device is available may be different for each device due to the differences in the capability $M_0$ of the energy harvesting modules and the repetition value $N_T$. However, we limit ourselves to the IoT device with the same repetition number $N_T$ and storage capability $M_0$ with the same energy availability probability.

	For strategy ${S\{N_T\}}$, the time for which an IoT device remains in the \rm{ON} state after it toggles from the \rm{OFF} state is given by $T_{0}^{\rm{on}}\{N_T\}$, and the time for which it remains in the \rm{OFF} state after toggling from the \rm{ON} state is given by $T_{0}^{\rm{off}}\{N_T\}$. For notational simplicity, the cutoff value $\{N_T\}$ will be dropped wherever appropriate. % The cycles of ON and OFF times go on as shown in Fig. 3. 
	It is noticed that both $T_{0}^{\rm{on}}$ and $T_{0}^{\rm{off}}$ are general random variables. 
	We first obtain $T_{0}^{\rm{on}}$ as\cite{name2} 
	\begin{align}\label{T_i}
	{T_{0}^{\rm{on}}\{N_T\} = \inf \{ \tau :{\xi _0}(\tau ) = 0|{\xi _0}(0) = {N_T}\} },
	\end{align}
	where $\xi_0(\tau )$ denotes the current energy level of a {randomly chosen} IoT device at time $\tau$. For this setup, the energy availability of a {randomly chosen} IoT device depends only on the means of $T_{0}^{\rm{on}}$ and $T_{0}^{\rm{off}}$, which is shown as
	\begin{align}\label{avalibility 1}
	{\eta_0} = \displaystyle \frac{{{\mathbb E}[{T_{0}^{\rm{on}}}]}}{{{\mathbb E}[{T_{0}^{\rm{on}}}] + {\mathbb E}[{T_{0}^{\rm{off}}}]}} = \displaystyle  \frac{1}{{1 + {{{\mathbb E}[{T_{0}^{\rm{off}}}]}}/{{{\mathbb E}[{T_{0}^{\rm{on}}}]}}}},
	\end{align}
	and $\mathbb E$$[{T_{0}^{\rm{on}}}]$ is the mean time the randomly chosen IoT device remains in the ON state, and $\mathbb E$$[{T_{0}^{\rm{off}}}]$ is the mean time it remains in the \rm{OFF} state. 
	\begin{proof}
		Let $\{ T_{0}^{\rm{on}}(k)\}$ and $\{ T_{0}^{\rm{off}}(k)\}$ be the sequences of the $k$th cycle of ON and OFF times, respectively. The availability can now be described by the fraction of time the {randomly chosen} IoT device remains in the ON  state as
		\begin{align}\label{avalibility 2}
		{\eta_0} = \mathop {\lim }\limits_{K \to \infty } \displaystyle \frac{{\sum\nolimits_{k = 1}^K T_{0}^{\rm{on}}(k)}}{{\sum\nolimits_{k = 1}^K T_{0}^{\rm{on}}(k) + \sum\nolimits_{k = 1}^K T_{0}^{\rm{off}}(k)}}.
		\end{align}
		Dividing both the numerator and the denominator by $K$ and invoking the law of large numbers, we have the result in (\ref{avalibility 1}).\end{proof}

	Now we need to calculate the mean ON time $\mathbb E$$[{T_{0}^{\rm{on}}}]$ and the mean \rm{OFF} time $\mathbb E$$[{T_{0}^{\rm{off}}}]$. For \rm{OFF} time, according to \cite{pishro2016introduction}, we have $\mathbb E[{T_{0}^{\rm{off}}}] = {{{N_T}}}/{{{\mu_0}}}$ (i.e., the time required to harvest $N_T$ units of energy), which is the sum of $N_T$ exponentially
	distributed random variables, each with mean $1/\mu_0$. 
	%It is known that the \rm{OFF} time $\mathbb E[{T_{i}^{\rm{off}}}] = \displaystyle \frac{{{N_T}}}{{{\mu_i}}}$ (i.e., the time required to harvest $N_T$ units of energy), which is the sum of $N_T$ exponentially distributed random variables, each with mean $1/\mu_i$. 
	Substituting into (\ref{avalibility 1}), 
	we obtain 
	\begin{align}\label{avalibility 3}
	{\eta_0} =\displaystyle \frac{1}{{1 + {{{N_T}}}/({{{\mu _0}{\mathbb E}[{T_{0}^{\rm{on}}}]}})}}.
	\end{align}
	
	To derive the mean \rm{ON} time $\mathbb E[{T_{0}^{\rm{on}}}]$, we first need to obtain the transmission matrix $Q_0$ for the birth-death process corresponding to the {randomly chosen} IoT device. According to the Kolmogorov differential equations \cite{karlin2014first}, $Q_0$=
	%\begin{align} \begin{bmatrix}
	%-\mu_D   &    \mu_D    &       0      &      \cdots   &   0    &     0\\
	%\nu _D   &  -\mu _D- \nu _D   &   \mu_D    &       \cdots    &  0     &    0 \\
	%0     &    \nu_D    &     - \mu_D-\nu _D  &   \cdots  &    0    &     0\\
	%\vdots    &      \vdots      &        \vdots      &      \ddots    &   \vdots  &    \vdots \\
	%0    &     0          &     0      &     \cdots   &   \nu _D  &   -\nu _D \end{bmatrix},\end{end}
	\begin{align}\label{generate}
	\begin{bmatrix}
	-\mu _0  &  \mu _0  &    0       & \cdots   &   0    &   0   &   0 \\
	\nu _0   &-\mu _0-\nu _0  &  \mu _0& \cdots &   0    &   0 &   0  \\
	0&\nu _0&- \mu _0-\nu _0&\cdots &   0    &   0 &   0  \\
	\vdots    &      \vdots &      \vdots     &            \ddots    &   \vdots  &    \vdots&    \vdots \\
	0    &     0          &     0      &     \cdots     & {\nu _0}& - \mu _0-\nu _0 &\mu _0\\
	0    &     0          &     0      &    \cdots &    0&\nu _0&- {\nu _0}
	\end{bmatrix},
	\end{align}
	where the first column corresponds to the energy state 0 and the states are in the ascending order.
	Then we obtain the following \textbf{Lemma 1.} from \cite{name2}. 
	
	\begin{lemma}
		\textnormal{(Mean Hitting Time).} 
		%If the embedded discrete Markov chain of the CTMC is irreducible 
		The expected hitting time of state $1$ (energy level $0$) starting from state m+$1$ (energy level m $\ne 0$) is
		\begin{align}\label{ON TIME}
		{\mathbb E}[{T_{0}^{\rm{on}}}(m)] = \left( {{{( - {{B}_0})}^{ - 1}}{\mathbbm 1})} \right)(m),
		\end{align}
		where $B_0$ is a defined $M_0 \times M_0$ matrix to be the restriction of matrix $Q_0$ to the set $M\setminus\{0\}$ i.e., $B_0=(Q_0(m,n),m\ne 0$, n $\ne 0)$, $-B_0$ is invertible and $\mathbbm 1$ is a column vector of all $1$'s. 
		
		% an defined $M_i \times M_i$ sub-matrix of $Q_i$ by deleting the first row and column of $Q_i$, i.e., $B_i=(Q_i(m,n),m\ne 0$, n $\ne 0)$ to be the restriction of matrix $Q_i$ to the set $M\setminus\{0\}$\ .
	\end{lemma}
	
	Now we can obtain a closed-form expression for the $(m, n)^{th}$ element in $(-B_0)^{-1}$ by
	\begin{align}\label{hit time}
	{( - {B_0})^{ - 1}}(m,n) = \frac{1}{{\nu _0^n}}\sum\limits_{k = 1}^{\min (m,n)} {\mu _0^{n - k}} \nu _0^{k - 1}.
	\end{align}
	\begin{proof}
		See Appendix A.
	\end{proof}
	
	Substituting (\ref{hit time}) into (\ref{ON TIME}) and then plugging into $m=N_T$, we have the mean \rm{ON} time for strategy $S\{N_T\}$ with energy storage capacity {$M_0$} as
	\begin{align}\label{ON TIME 2}
	&{\mathbb E}[T_{0}^{\rm{on}}({N_T})] \\ \nonumber
	&=\displaystyle\frac{{{{\Big( \displaystyle\frac{{{\mu _0^{}}}}{{\nu _0^{}}} \Big)}^{M_0 + 1}}\Big( {1 - {{\Big( \displaystyle\frac{{{\mu _0^{}}}}{{\nu _0^{}}} \Big)}^{ - {N_T}}}} \Big)}}{{\Big( {1 -  \displaystyle\frac{{{\mu _0^{}}}}{{\nu _0^{}}} } \Big)}} \displaystyle\frac{1}{{\nu _0^{} - \mu _0^{}}} 
	- \frac{{{N_T}}}{{\mu _0^{} - \nu _0^{}}}.
	\end{align}
	
	Submitting (\ref{ON TIME 2}) into (\ref{avalibility 3}), we obtain the energy availability in the following \textbf{Theorem 1.}.
	\begin{theorem}
		\textnormal{(Energy Availability).} The energy availability of a {randomly chosen} IoT device  is given by
		\begin{align}\label{final AVA}
		{\eta _0} %\nonumber\\ 
		= \displaystyle\frac{1}{{1 + \displaystyle\frac{{{N_T}{{\Big(1- \displaystyle\frac{{{\mu _0^{}}}}{{\nu _0^{}}} \Big)}^2}}}{{{{\Big( \displaystyle\frac{{{\mu _0^{}}}}{{\nu _0^{}}} \Big)}^{{M_0} + 2}}\Big( {1 - {{\Big( \displaystyle\frac{{{\mu _0^{}}}}{{\nu _0^{}}} \Big)}^{ - {N_T}}}} \Big) + \displaystyle\frac{{{\mu _0^{}}}}{{\nu _0^{}}} \Big(1- \displaystyle\frac{{{\mu _0^{}}}}{{\nu _0^{}}} \Big){N_T}}}}}.
		\end{align}
		%where ${M_0}^\star=M_0-M_{0}^c+N_T$.
	\end{theorem}
	In (\ref{final AVA}), it can be shown that the energy availability increases with increasing $\mu_0$ and {$M_0$}. For illustration, the relationships between $\eta_0$ and $N_T$, {$M_0$}, $\mu_0$ are analyzed in Section V.

	\section{RACH Transmission Success Probability} 
	In the NB-IoT repetition scheme, an active IoT device will repeat the same preamble $N_T$ times (i.e., the dedicated repetition value). In each repetition, a preamble is composed of four symbol groups transmitted without gaps, where the first preamble symbol group is transmitted via a sub-carrier determined by pseudo-random hopping (i.e., the hopping depends on the current repetition time and the Narrowband physical Cell ID, a.k.a NCellID\cite{schlienz2016narrowband}), and the following three preamble symbol groups are transmitted via sub-carriers determined by the fixed size frequency hopping\cite{Tel2016}. This frequency hopping algorithm is designed in a way that different selections of the first subcarrier lead to hopping schemes which never overlap. %{\color{red}
	%Hence there are as many different congestion free preambles as there are subcarrier allocated to the NPRACH.} 
	Specifically, if two or more IoT devices chose the same first sub-carrier in a single RAO, the following sub-carriers (i.e., in the same RAO) would be same, due to that these two hopping algorithms lead to one-to-one correspondences between the first sub-carrier and the following sub-carriers (i.e., these IoT devices either collide on the full set or not collide at all in a single RAO). In this setup, the RACH success refers to the preamble being successfully transmitted to the associated eNB (i.e., received SINR is greater than the SINR threshold) and no collision occurs (i.e., no other IoT devices successfully transmit a same preamble to the typical eNB simultaneously). 
	
	For an IoT device to be able to initiate an uplink transmission, the energy harvested by this device should be sufficient to perform RACH and data transmission. We have defined the energy availability of a {randomly chosen} IoT device after harvesting enough energy as $\eta _0$ in \textbf{Theorem 1}. In this section, we first formulate the SINR outage condition to drive the preamble transmission success probability and then facilitate the analysis of the  RACH success probability.
	
	\subsection{SINR Definition} 
	
	Recall that  each IoT device transmits a {randomly chosen} preamble to its associated eNB to request for channel resources, where different preambles represent orthogonal sub-channels, and thus only IoT devices choosing the same preamble have correlations.
	The RACH analysis in this work needs to take into account both the inter- and intra-cell interference\footnote{
		In LTE, the PRACH root sequence planning is used to mitigate inter-cell interference among neighboring  BSs (i.e., neighboring BSs could use different roots to generate preambles) \cite{dahlman20134g}. However, as in \cite{jiang2017random}\cite{7917340}, we focus on providing a general analytical framework of cellular networks considering both the inter- and intra-interference without using PRACH root sequence planning.
		That is to say, we consider intra-cell interference due to the fact that the IoT devices in the same cell  may choose the same preamble and we consider the inter-cell interference due to the fact that the IoT devices in different cells share the preamble sequence pool among eNBs. 
		We also obtain  the  approximation results of the networks with perfect PRACH root sequence planning (e.g., no inter-cell interference from any cells) in Lemma. 3.  But the extension taking into account  PRACH root sequence planning will be treated in future works. 
		%Similar as \cite{8258982}\cite{7917340}, we focus on providing a general analytical framework of cellular network,  considering both the inter- and intra-interference.
	}.
	The received power at a typical eNB from a {randomly chosen} IoT device of interest is therefore
	$P_{0}=Ph_0r_0^{-\alpha}$, where $h_0$ and $r_0$ are the channel power gain and the distance from the typical IoT device to its associated eNB respectively.
	Using the received power over the link of interest and the interference power, the SINR received at the typical eNB at the origin can be written as
	%\begin{align}\label{SINR}
	%{\rm{SINR}}({r_i}) = \frac{{P{h_i}r_i^{ - \alpha }}}{{\mathcal {I}_i^{{\mathop{ intra}} } + {\sigma ^2}}},
	%\end{align}
	{
		\begin{align}\label{SINR}
		{\rm{SINR}}({r_0}) = \frac{{P{h_0}r_0^{ - \alpha }}}{{\mathcal {I}_0^{{\mathop{ intra}} } + \mathcal {I}_0^{{\mathop{ inter}} }+{\sigma ^2}}}
		=\frac{{P{h_0}r_0^{ - \alpha }}}{{\mathcal {I}_0  +{\sigma ^2}}},
		\end{align}
		where $\sigma ^2$ is noise power, and $\mathcal {I}_0$ is aggregate  interference of the typical IoT device with}
	% \begin{align}\label{INTRA}
	% \mathcal {I}_i^{{\mathop{ intra}} } = \sum\limits_{j \in {\mathcal {Z}_{{\mathop{ intra}} }}}^{} {P{h_j}{r_j}^{ - \alpha }}. 
	% \end{align}
	\begin{align}\label{INTRA}
	\mathcal {I}_0^{{\mathop{}} } = \sum\limits_{j \in {\mathcal {Z}_{{\mathop{ D}} }}}^{} {P{h_j}{r_j}^{ - \alpha }}. 
	\end{align}
	% \begin{align}\label{INTRA}
	% \mathcal {I}_0^{{\mathop{ inter}} } = \sum\limits_{j \in {\mathcal {Z}_{{\mathop{ inter}} }}}^{} {P{h_j}{r_j}^{ - \alpha }}. 
	% \end{align}
	
	In (\ref{INTRA}),  $h_j$ and $r_j$ are channel power gain and the distance from the interfering IoT devices to the typical eNB, $P$ is the transmission power of the IoT devices, and $\mathcal Z_{D}$ is the set of interfering IoT devices for the typical IoT device. 
	We note that only the active IoT devices choosing the same preamble will generate interference.
	The  density of active IoT devices choosing the same preamble is obtained as follows.
	
	\begin{remark}\label{remarkdensity}
		\textnormal{(The density of active IoT devices choosing the same preamble)}. Note that inactive IoT devices (those without enough energy or data packets in buffer) do not attempt RACH, such that they do not generate interference. According to the repetition scheme mentioned earlier, each active IoT device will contend on all $L=48$ sub-carriers due to the single repetition value configuration, and thus each preamble has an equal probability $(1/L)$ to be chosen. As only active IoT devices will try to request uplink channel resources, we define the non-empty data packets probability of each IoT device $\mathcal{A}_a \in [0, 1]$ following a Bernoulli process. We also have the energy availability $\eta_0$ of each IoT device, then according to the thinning process \cite{kingman1993poisson}, the density of active IoT devices choosing the same preamble can be expressed as
		\begin{align}\label{density}
		{\lambda _{Da}} = \mathcal{A}_a{\eta _0}{\lambda _D}/L.
		\end{align}
	\end{remark}

	\subsection{RACH Success Probability}
	{
		We formulate the RACH success probability under both SINR outage and collision conditions.
		We perform the analysis on an eNB associating with a {randomly chosen} active IoT device in terms of the RACH success probability. The RACH success refers to the preamble being successfully transmitted to the associated eNB (i.e., received SINR is greater than the SINR threshold) and no collision occurs (i.e., no other IoT devices successfully transmits a same preamble to the typical eNB simultaneously). 
	}
	First, we formulate the SINR outage condition. The typical IoT device transmits a preamble successfully if any repetition successes, and in a single repetition, a preamble is successfully received at the associated eNB if its all four received SINRs are above the SINR threshold $\gamma_{th}$. %According to the frequency hopping algorithm, these intra-cell IoT devices either collide on the full set or not collide at all in a RAO.
	%according to the frequency hopping algorithm, if two or more IoT devices chose the same first sub-carrier in a RAO, the following sub-carriers in the same RAO would be same. That is to say, these intra-cell IoT devices either collide on the full set or not collide at all in a RAO. 
	Thus, the preamble transmission success probability of a {randomly  chosen} IoT device under $N_T$ repetitions  is expressed as
	% \begin{align}\label{preamble1}
	% {{\mathbb P}_{S,i}}[{N_T},n] = 1 - \underbrace {\prod\limits_{{n_T} = 1}^{{N_T}} {\Big( {1 - \underbrace {{{\mathbb P}_i}[{\theta _{{n_T}}}({r_i})| {N = n} ]}_{\rm I}} \Big)} }_{{\rm I}{\rm I}},
	% \end{align}
	\begin{align}\label{preamble1}
	{{\mathbb P}_{S,0}}[{N_T}] = 1 - \underbrace {\prod\limits_{{n_T} = 1}^{{N_T}} {\Big( {1 - \underbrace {{{\mathbb P}_0}[{\theta _{{n_T}}}({r_0})|r_0]}_{\rm I}} \Big)} }_{{\rm I}{\rm I}}.
	\end{align}
	
	I is the probability that all four (a preamble consists of four preamble symbol groups) time-correlated preamble symbol groups in the $n_T$th repetition are successfully transmitted,
	II is the probability that all $N_T$ repetitions of a preamble transmission are failed, and
	% {
	% \begin{align}\label{FOUR SINR}
	% {\theta _{{n_T}}}({r_0}) = \big\{&{\rm{SINR}}_{{n_T},k}({r_0}) \ge {\gamma _{th}}, k=1,2,3 \ and\  4 \big\}.
	% \end{align}
	% }
	\begin{align}\label{FOUR SINR}
	{\theta _{{n_T}}}({r_0}) = \big\{&{\rm{SINR}}_{{n_T},1}({r_0}) \ge {\gamma _{th}}, {\rm{SINR}}_{{n_T},2}({r_0}) \ge {\gamma _{th}},\nonumber \\  &{\rm{SINR}}_{{n_T},3}({r_0}) \ge {\gamma _{th}}, {\rm{SINR}}_{{n_T},4}({r_0}) \ge {\gamma _{th}} \big\}.
	\end{align}
	In \eqref{FOUR SINR}, $\gamma _{th}$ is the SINR threshold, and
	SINR$_{{n_T},1}({r_0})$, SINR$_{{n_T},2}({r_0})$, SINR$_{{n_T},3}({r_0})$ and SINR$_{{n_T},4}({r_0})$ are the received SINRs of the four symbol groups in the $n_T$th repetition of the typical IoT device.
	Based on the Binomial theorem, the preamble transmission success probability in (\ref{preamble1}) can be rewritten as
	\begin{align}\label{SI}
	&{{\mathbb{P}}_{S,0}}[{N_T}]=\nonumber \\  &\sum\limits_{{n_T} = 1}^{{N_T}} {{( - 1)}^{{n_T} + 1}} {\Big( \begin{array}{l}
		{N_T}\\{n_T}\end{array}  \Big)}{{\mathbb{P}}_0}[{\theta _1}({r_0}),{\theta _2}({r_0}), \cdots ,{\theta _{{n_T}}}({r_0})|r_0] ,
	\end{align}
	where $\bigg( \begin{array}{l}
	{N_T}\\
	{n_T}
	\end{array} \bigg) =\displaystyle \frac{{{N_T}!}}{{{n_T}!( {{N_T} - {n_T}} )!}}$ is the binomial coefficient, and $\displaystyle{{\mathbb P}_0[{\theta _1}({r_0}),{\theta _2}({r_0}), \cdots, {\theta _{{n_T}}({r_0})}]}$ is the probability that all of $4\times n_T$ (a preamble consists of four preamble symbol groups) time-correlated preamble symbol groups are successfully transmitted.

	We note that the preamble transmission success probability in \eqref{SI} depends on the transmission distance $r_0$.
	Taking into account that each IoT device associates to its geographically nearest eNB, $r_0$ is the minimum distance between the eNB and the typical IoT device.
	The PDF of the shortest distance between any point BS and the IoT device with radius $r_0$ is \cite{7852435}
		\begin{align}\label{r_0}
		{f_{{R_0}}}({r_0}) \approx 2\varepsilon\pi{\lambda_B}r_0\exp(-\varepsilon{\lambda_B}{\pi}r_0^2),
		\end{align}
		where $\varepsilon=1$ when $\lambda_{Da}\ll\lambda_B$ and $\varepsilon=1.25$ when $\lambda_{Da}\gg\lambda_B$.
	%We derive the PDF of the shortest distance $r_0$ in the following \textbf{Lemma.}.

	For ease of presentation, we set $l = 4\times n_T$, and the probability that all of $4\times n_T$  preamble symbol groups are successfully transmitted is presented in the following \textbf{Lemma 2.} 
		\begin{lemma}
			The probability that all of $4\times n_T$ received {\rm{SINRs}} at the eNB from a {randomly chosen} IoT device exceed a certain threshold $\gamma_{th}$ is expressed as
			\begin{align}\label{conditionr}
			&p_0({\gamma _{th}}) = {\mathbb E_{{R_0}}}\Big[{\mathbb P}_0[ {{\theta _1}(r_0),{\theta _2}(r_0),...,{\theta _{{n_T}}}(r_0)| {{r_0}}} ]\Big]
			%& = \int\limits_{{D_{i - 1}}}^{{D_i}} {{\rm P}_k^i\left[ {{\theta _1}(k),{\theta _1}(k),...,{\theta _{{m_i}}}(k)\left| {{r_i}} \right.} \right]} {f_{{R_i}}}({r_i})d{r_i}
			\nonumber\\
			&= \int{{\mathbb P}_0[ {{\theta _1}(r_0),{\theta _2}(r_0),...,{\theta _{{n_T}}}(r_0)| {{r_0}}} ]} f(r_0)d{r_0}\nonumber \\			
			&=\int_0^\infty  {} 2\varepsilon\pi{\lambda_B}r_0\exp\Big(-{\varepsilon\lambda_B}{\pi}r_0^2{{-\frac{{{l\gamma _{th}\sigma ^2}r_0^\alpha }}{P}}}\Big)\nonumber \\			
			&\times\exp \Big( \overset{\text{}}  { - 2\pi {\lambda _{Da}}\int_{0}^\infty  {\Big[ {1 - \big({}{{1 + {\gamma _{th}}r_0^\alpha{y^{ - \alpha }}}}\big)^{-l}} \Big]} ydy} \Big)d{r_0}.
			\end{align}
			%where (a) follows from the independence of $h_{0}^l$, and $\mathcal I_{0}$ is given in \eqref{INTRA}.
		\end{lemma}
		\begin{proof}
			See Appendix B.
		\end{proof}

	%According to the frequency hopping algorithm, these intra-cell IoT devices either collide on the full set or not collide at all in a RAO.
	%Recall that a collision occurs if an eNB receives multiple preambles from the same set of sub-carriers at the same time. We derive the RACH success probability of the $i$th IoT device in following theorem.
	Next, we formulate the RACH  success probability taking into account both the SINR outage and the collision. 
	The RACH success probability  is represented in the following \textbf{Theorem 2.}
	
	\begin{theorem}
		In the energy harvesting NB-IoT network, the RACH success probability of a {randomly chosen} IoT device  is derived as
		\begin{align}\label{rach_1}
		{{\cal P}_0}&={{\mathbb{E}}_N}{\Big[ { {{{\mathbb P}_{S,0}}[{N_T}]}{\prod\limits_{j = 1}^{n} {\Big(1 - {{\mathbb P}_{S,j}}[{N_T}]\Big)}\Big|N=n }} \Big]}\nonumber\\
		&=\sum\limits_{n = 0}^\infty\Big\{\underbrace{{\mathbb P}[N = n]}_{{\rm I}} { {\underbrace {{{\mathbb P}_{S,0}}[{N_T}]}_{{\rm I}{\rm I}}\underbrace {\prod\limits_{j = 1}^{n} {\Big(1 - {{\mathbb P}_{S,j}}[{N_T}]\Big)}\Big|N=n }_{{\rm I}{\rm I}{\rm I}}} \Big\}},
		\end{align}
		where
		\begin{align}\label{N=n}
		\mathbb{P}[N = n] = \displaystyle\frac{{{c^{(c + 1)}}\Gamma (n + c + 1){{\big( {{{{\lambda _{Da}}}}/{{{\lambda _B}}}} \big)}^n}}}{{\Gamma (c + 1)\Gamma (n + 1){{\big( {{{{\lambda _{Da}}}}/{{{\lambda _B}}} + c} \big)}^{n + c + 1}}}},
		\end{align}
		and
		\begin{align}\label{PRE_1}
		&{{{\mathbb P}_{S,0}}[{N_T}]}
		= \sum\limits_{{n_T} = 1}^{{N_T}} {{( - 1)}^{{n_T} + 1}} {\Big( \begin{array}{l}
			{N_T}\\{n_T}\end{array}  \Big)}\nonumber\\
		&\times\int_0^\infty  {} 2\varepsilon\pi{\lambda_B}r_0\exp\Big(-\varepsilon{\lambda_B}{\pi}r_0^2{{-\frac{{{l\gamma _{th}\sigma ^2}r_0^\alpha }}{P}}}\Big)\nonumber\\
		&\times\exp \Big( \overset{\text{}}  { - 2\pi {\lambda _{Da}}\int_{0}^\infty  {\Big[ {1 - \big({}{{1 + {\gamma _{th}}r_0^\alpha{y^{ - \alpha }}}}\big)^{-l}} \Big]} ydy} \Big)d{r_0}.
		\end{align}
		Part I the Probability Mass Function (PMF) of the number of intra-cell interfering\footnote{
			We derive the PMF of the number ($N$) of the other interfering IoT devices in the Voronoi cell to which a randomly chosen IoT device belongs, i.e., there are $n$ interfering IoT devices ($n+1$  IoT devices) in one cell. According to the Slivnyak’s Theorem \cite{haenggi2012stochastic}\cite{yu2013downlink}, the locations of interfering IoT devices follow the Palm distribution of $\Phi_{D_a}$, which is the same as the original $\Phi_{D_a}$.}  IoT devices for a typical BS $N=n$  derived following \cite[Eq.(3)]{yu2013downlink},
		where $c = 3.575$ is a constant related to the approximate PMF of the PPP Voronoi cell and $\Gamma (\cdot)$ is the gamma function.
		Part II is the preamble transmission success probability of a {randomly chosen} IoT device obtained by substituting (\ref{conditionr}) into (\ref{SI}).
		Part III is the preamble transmission failure probability that the transmissions from other $n$ intra-cell interfering IoT devices are not successfully received by the BS, i.e., the non-collision probability of the typical IoT device conditioning on $n$.
	\end{theorem}

	\begin{remark}
		It is evident from \eqref{PRE_1} that the transmission success probability (II in \eqref{rach_1}) of the typical IoT device increases, whereas the non-collision probability (III in \eqref{rach_1}) decreases with increasing the repetition value $N_T$ and decreasing the the received {\rm SINR} threshold $\gamma_{\rm th}$.
		%and that the transmission success probability (II in \eqref{rach_1}) of the typical IoT device increases, whereas the non-collision probability (III in \eqref{rach_1}) decreases with increasing the repetition value $N_T$ and the received {\rm SINR} threshold $\gamma_{\rm th}$.
		%and the transmission success probability (II in \eqref{rach_1})  is inversely proportional to the received {\rm SINR} threshold $\gamma_{\rm th}$, whereas the non-collision probability (III in \eqref{rach_1}) (III in \eqref{rach_1})  are directly proportional to the received SINR threshold. 
		Therefore, there exists a tradeoff between transmission success probability and non-collision probability.
		For illustration, the relationship among RACH access success probability $(\mathcal{P}_0)$, the transmission success probability $(\mathcal{P}_0$ with III=1), and the non-collision probability $(\mathcal{P}_0$ with II=1) versus repetition value $N_T$ and the received {\rm SINR} threshold $\gamma_{\rm th}$ is shown in Fig. 4.
	\end{remark}

	\begin{figure}[htbp!]
		\begin{center}
			\begin{minipage}[t]{0.48\textwidth}
				\centering
				\includegraphics[width=3.0in,height=2.4in]{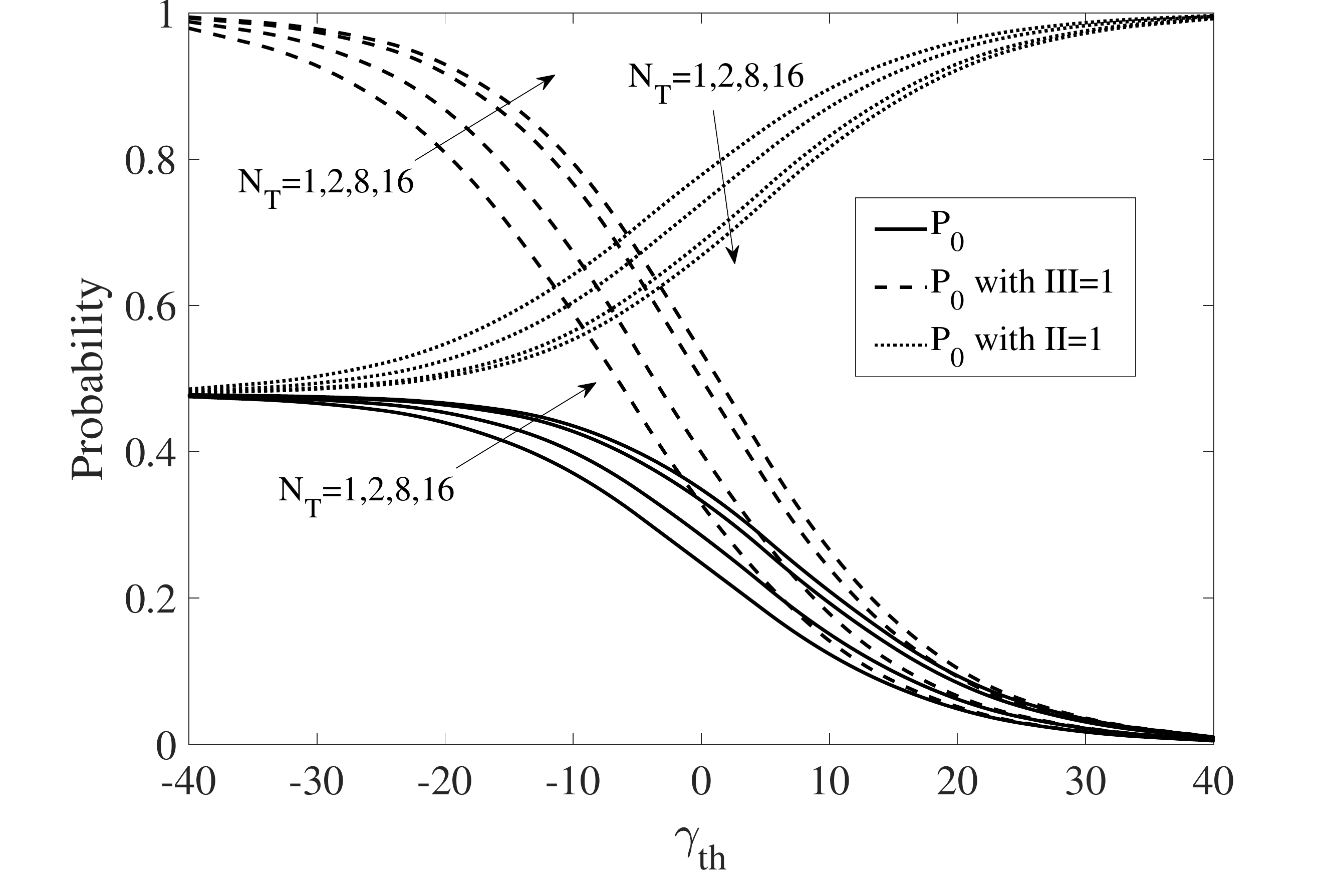}
			%	\vspace*{-0.4cm}
				\caption{\scriptsize {Comparing RACH access success probability (${\cal P}_0$), transmission success probability (${\cal P}_0$ with III = 1), and non-collision probability (${\cal P}_0$ with II = 1), $\lambda_D=10^4$ UEs/km$^2$.}}
				%Comparing  ${\cal P}_0$, ${\cal P}_0$ with III = 1,  and ${\cal P}_0$ with II = 1. The parameters are $\lambda_B=0.1$ BS/km$^$2, $\lambda_D=10^4$ UEs/km$^$2, ${\cal A}_a=0.001$, $\eta_0=0.3$, $P=20$ mw, and $\sigma^2= −138.3$ dBm..}}
			\end{minipage}
			\label{fig:4}			
		\end{center}
	\end{figure}

	Note that the RACH  success probability of a randomly chosen IoT device in networks with perfect PRACH root sequence planning could be obtained by only considering the intra-cell interference.
	Considering that the practical Voronoi cells do not have a constant radius, we use the average radius $D= 1/\sqrt {\pi {\lambda _B}}$\cite{novlan2013analytical} to approximate it\footnote{Note that for the networks with perfect PRACH root sequence planning, it is difficult to characterize the radius of the practical cell. We use $D= 1/\sqrt {\pi {\lambda _B}}$ to approximate the radius of the practical cell and the results depend on the practical cell shape. 
		Our results have a good match when choosing a proper deployment area and $\lambda_{\rm{B}}$ as shown in Fig. 10 and Fig. 12.}. Thus, the RACH success probability is given in the following \textbf{Lemma 3}.

	\begin{lemma}
		The RACH success probability of a {randomly chosen} IoT device in NB-IoT networks with perfect PRACH root sequence planning  is derived as
		\begin{align}\label{rach_INTRA}
		{{\cal P}_0}&={{\mathbb{E}}_N}{\Big[ { {{{\mathbb P}_{S,0}}[{N_T}]}{\prod\limits_{j = 1}^{n} {\Big(1 - {{\mathbb P}_{S,j}}[{N_T}]\Big)}\Big|N=n }} \Big]}\nonumber\\
		&=\sum\limits_{n = 0}^\infty\Big\{\underbrace{{\mathbb P}[N = n]}_{{\rm I}} { {\underbrace {{{\mathbb P}_{S,0}}[{N_T}]}_{{\rm I}{\rm I}}\underbrace {\prod\limits_{j = 1}^{n} {\Big(1 - {{\mathbb P}_{S,j}}[{N_T}]\Big)}\Big|N=n }_{{\rm I}{\rm I}{\rm I}}} \Big\}},
		\end{align}
		where $\mathbb{P}[N = n]$ is given in \eqref{N=n}
		and
		\begin{align}\label{PRE__INTRA}
		&{{{\mathbb P}_{S,0}}[{N_T}]}
		= \sum\limits_{{n_T} = 1}^{{N_T}} {{( - 1)}^{{n_T} + 1}} {\Big( \begin{array}{l}
			{N_T}\\{n_T}\end{array}  \Big)}\nonumber\\
		&\times\int_0^\infty  {} 2\varepsilon\pi{\lambda_B}r_0\exp\Big(-\varepsilon{\lambda_B}{\pi}r_0^2{{-\frac{{{l\gamma _{th}\sigma ^2}r_0^\alpha }}{P}}}\Big)\nonumber\\
		&\times\exp\Big({ - 2\pi {\lambda _{Da}}\displaystyle\int_{0}^{\frac{1}{\sqrt {\pi {\lambda _B}}}}  {\Big[ {1 - \Big({}{{1 + {\gamma _{th}}r_0^\alpha{y^{ - \alpha }}}}\Big)^{-l}} \Big]} ydy} \Big)d{r_0}.
		\end{align}
	\end{lemma}

	\section{Simulation and Discussion}
	{
		In this section, we verify our analytical results by comparing the theoretical RACH success probabilities with the results from Monte-Carlo simulations.
		For numerical verification, we compute the RACH success probability from Monte-Carlo simulations as follows. We simulate the spatial model described in Section II in MATLAB. 
		The eNBs and IoT devices are deployed via independent HPPPs in a $2\times 10^4$ km$^2$ circle area. Each IoT device associated with its nearest eNB. 
		The IoT devices and the eNBs  remain spatially static during a TTI.
		The channel fading gains between the IoT devices and eNBs are modeled by exponentially distributed random variables.
		%{To built a general model with less complexity and high scalability for RACH analysis, we assumed $E_{i0}^{DA}=0$. The unfixed energy consumption for data transmission in each time slot can be a good extension of this work.}
		Unless otherwise stated, we set $P=0.02$ w, $T_r=6$ ms, $T_g=31$ ms, $\mu_i=0.05$, $\lambda_{\rm{B}}=0.1$ eNBs/km$^2$, $\lambda_{\rm{D}}=10^2$ IoT devices/km$^2$, $\gamma_{th}= 20$ dB, $\alpha = 4$, $\mathcal{A}_a=0.001$, the bandwidth of a subcarrier is BW$ = 3.75$ kHz, and thus the noise is $\sigma^2 = −174+ 10$log$10$(BW) $= −138.3$ dBm.
		For each realization of this setup, the uplink communication has declared a success if 1) the calculated SINR exceeds a pre-determined threshold $\gamma_{th}$ and 2) other uplink communications using the same preamble do not exceed $\gamma_{th}$. 
		% This experiment is then repeated for a large number of realizations ($10^4$) to obtain the empirical RACH success probability. 
		%These empirical results are then compared with the RACH success probability results evaluated from the analytical expressions in Lemma 3 to verify the accuracy of the analysis.
		In all figures of this section, “Analytical” and “Simulation” are abbreviated as “Ana.” and “Sim.”, respectively.}

	\begin{figure}[htbp!]
	\begin{center}			
		\begin{minipage}[t]{0.48\textwidth}
			\centering
			\includegraphics[width=3.0in,height=2.4in]{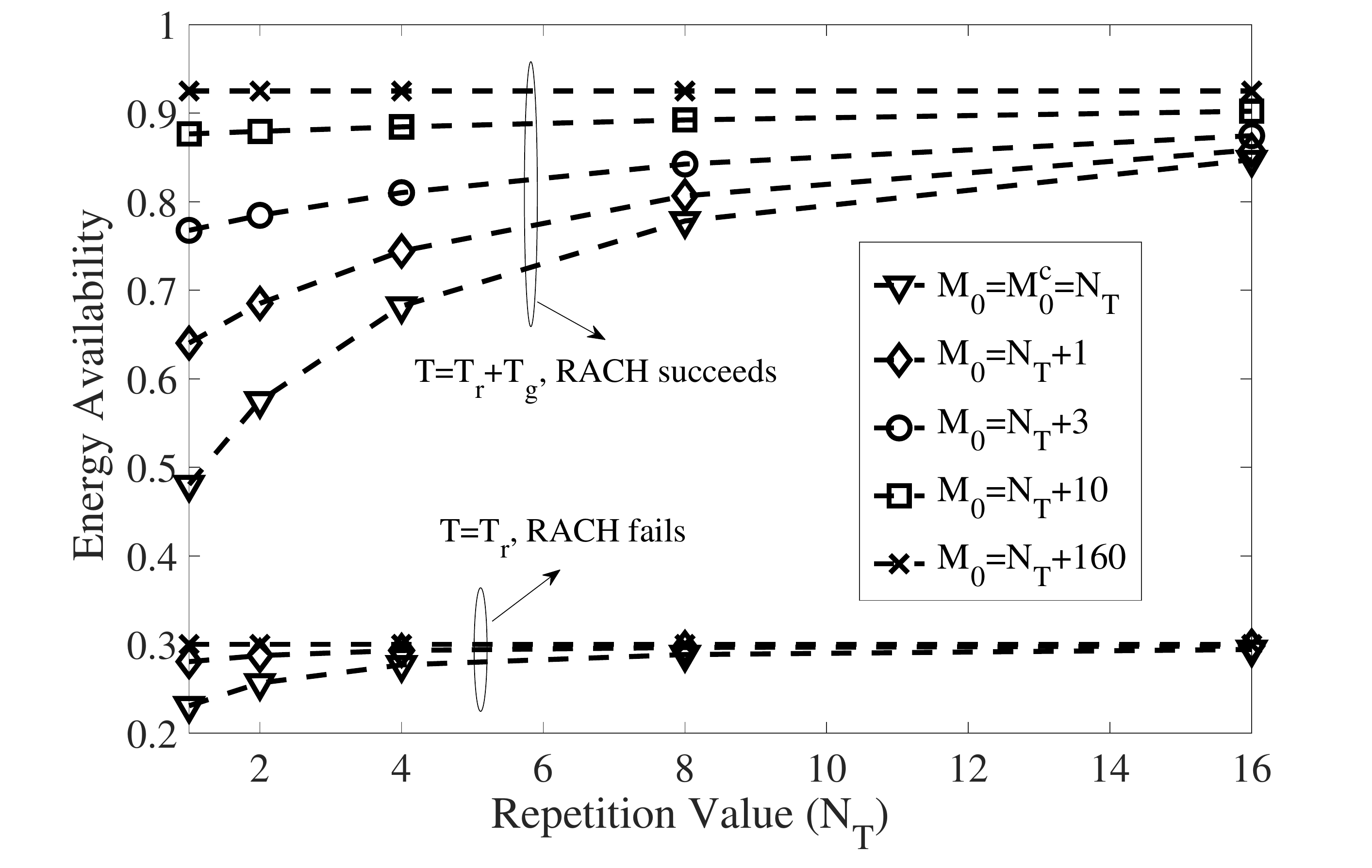}
			%\vspace*{-0.4cm}
			\caption{\scriptsize {The energy availability of the IoT device for $S\{N_T\}$ versus the repetition value $N_T$ for various storage capacities ${M_0}$.}}
		\end{minipage}
		\label{fig:5}
	\end{center}
\end{figure}
	
	%\begin{figure}
	%	\centering
	%	\includegraphics[width=3.5in,height=2.5in]{energynew11.eps}
	%	\caption{The energy availability of IoT device for $S\{N_T\}$ with storage capacity (${M_i}-{M_{ic}}+N_T$).}
	%	\label{fig:my_label}
	%\end{figure}
	
	%\begin{figure}
	%	\centering
	%	\includegraphics[width=3.5in,height=2.5in]{ALLPOWER.eps}
	%	\caption{RACH success probability with different repetition values and transmission power ($\eta_i=0.3$, $\lambda_{\rm{B}}=0.05$ eNBs/km$^2$, $\gamma_{th}= 0$ dB, $\alpha = 4$, and $\lambda_{\rm{D}}=10$ devices/km$^2$)}
	%	\label{fig:my_label}
	%\end{figure}

	\subsection{Analysis of the Energy Availability}
	{Fig. 5}
	plots the energy availability of a {randomly chosen} IoT device under our strategy  $S\{ N_T\}$ versus the preamble  repetition value $N_T$
	for various storage capacities {${M_0}$} using (\ref{final AVA}). 
	We assume that the RACH in each TTI succeeds or fails on the full set (i.e., the RACH always succeeds or all the RACH always fails), and then we derive the upper bound and the lower bound of the energy availability.
	As such, we could give the availability region for various values of energy availability.
	
	{We first observe that the energy availability of the IoT device increases with increasing the preamble repetition value {$N_T$ at first and then remains unchanged}. This is  due to the fact that for the same storage capacity, increasing the repetition value increases the operation time of the IoT device, i.e., the IoT device spends more time in \rm{ON} state. 
		In addition, since ${M_0^c}=N_T$, increasing the repetition value $N_T$, i.e., increasing the cutoff value, results in more time needed for the IoT device to harvest sufficient energy, i.e, stay in \rm{OFF} state for more time before the transmission.
		As such, to obtain a higher energy availability, a larger repetition value $N_T$ is needed, but if the repetition value is overestimated, the IoT device will waste the potential resource for data transmission and lead to lower resource efficiency. That is to say, the repetition value needs to be optimized. 
	}
	{
		Interestingly, we also observe that for energy storage capacity much larger than cutoff value, the energy availability approaches a specific value for different preamble repetition values, e.g., $\eta_0=0.3$ for the lower bound when and $\eta_0=0.92$ for the upper bound when $M_{0}-{M_{0}^c} \ge160$, which reveals that this setup is surprisingly reliable if it is designed properly, despite the randomness in the energy harvesting.}

	\begin{figure}[htbp!]
		\centering
		\includegraphics[width=3.6in,height=2.6in]{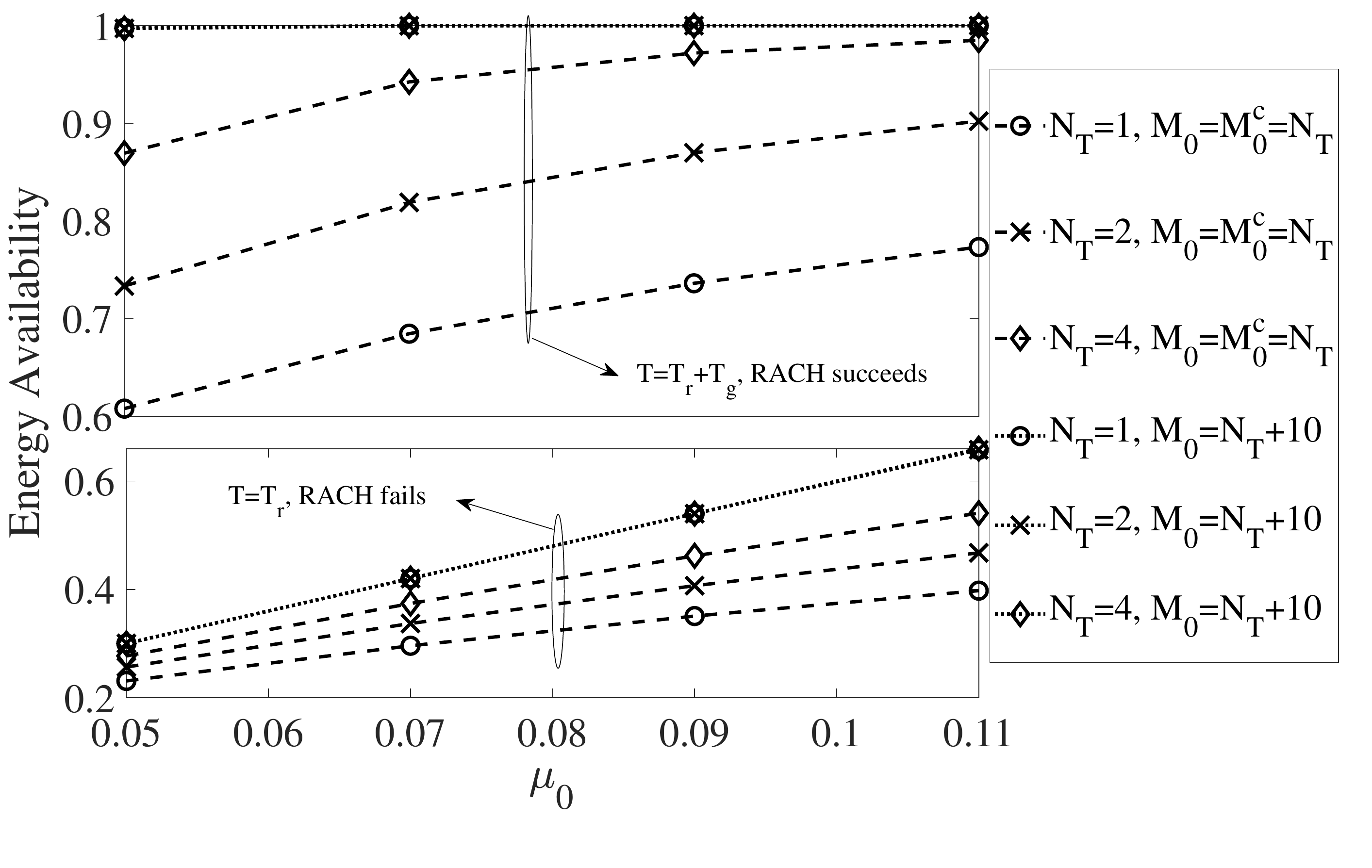}
		\caption{{The energy availability of the IoT device for $S\{N_T\}$ versus the  energy harvesting rate $\mu_0$ for various storage capacities ${M_0}$ and repetition values $N_T$.
		}}
		\label{fig:6}
	\end{figure}
	{Fig. 6 plots the energy availability of the {randomly chosen} IoT device under our strategy  $S\{ N_T\}$ versus the energy harvesting rate $\mu_0$
		for various storage capacities {${M_0}$} and repetition values $N_T$ using (\ref{final AVA}). 
		We first observe that the energy availability of the IoT device increases with increasing the energy harvesting rate $\mu_0$ when the storage capacity $M_0$ is not large enough, e.g.,$M_0=N_T$. This is due to the fact that for the same storage capacity and the cutoff value, increasing the energy harvesting rate results in less time needed for the IoT device to harvest sufficient energy, i.e, stay in \rm{OFF} state for less time before the transmission. 
	}
	{Interestingly, we also observe that for energy storage much larger than cutoff value, e.g., $M_0-M_0^c\ge10$, the energy availability approaches a specific value for different preamble repetition values. That is to say, if it is designed properly, the energy availability is independent of $N_T$.}

	% {According to \eqref{nu_i}, we also notice that when RACH succeeds, $T$ = $T_{r}$ + $T_{g}$, 
	% }
	% increasing the rate ratio $\mu_i/\nu_i$ of the IoT device increases the energy availability. This can be explained by the reason that the IoT device needs less time to harvest sufficient energy when it harvests energy faster, i.e., it stays in \rm{OFF} state for less time before harvesting sufficient energy.

	\subsection{Validation of the RACH success probability}
	In this section, we simulate an NB-IoT network and evaluate the RACH success probability based on the energy availability analyzed above.
	Fig. 7 plots the RACH success probability of a randomly chosen IoT device versus the energy availability $\eta$ for various preamble repetition values $N_T$.
	We observe that the RACH success probability deteriorates as the energy availability of IoT devices increases. This is due to the fact that increasing energy availability increases the number of active devices, which leads to lower received SINR and a higher probability of collision. In order to obtain fundamental insights on the RACH success probability due to the preamble repetition value $N_T$, in the following analysis, we use unchanged $\eta_0=0.3$ for different $N_T=$1, 2, 4, and 8, which is obtained by fine tuning $M_{0}-{M_{0}^c} \ge160$  as shown in Fig. 5.
	\begin{figure}[htbp!]
		\begin{center}
			\begin{minipage}[t]{0.48\textwidth}
				\centering
				\includegraphics[width=3.0in,height=2.4in]{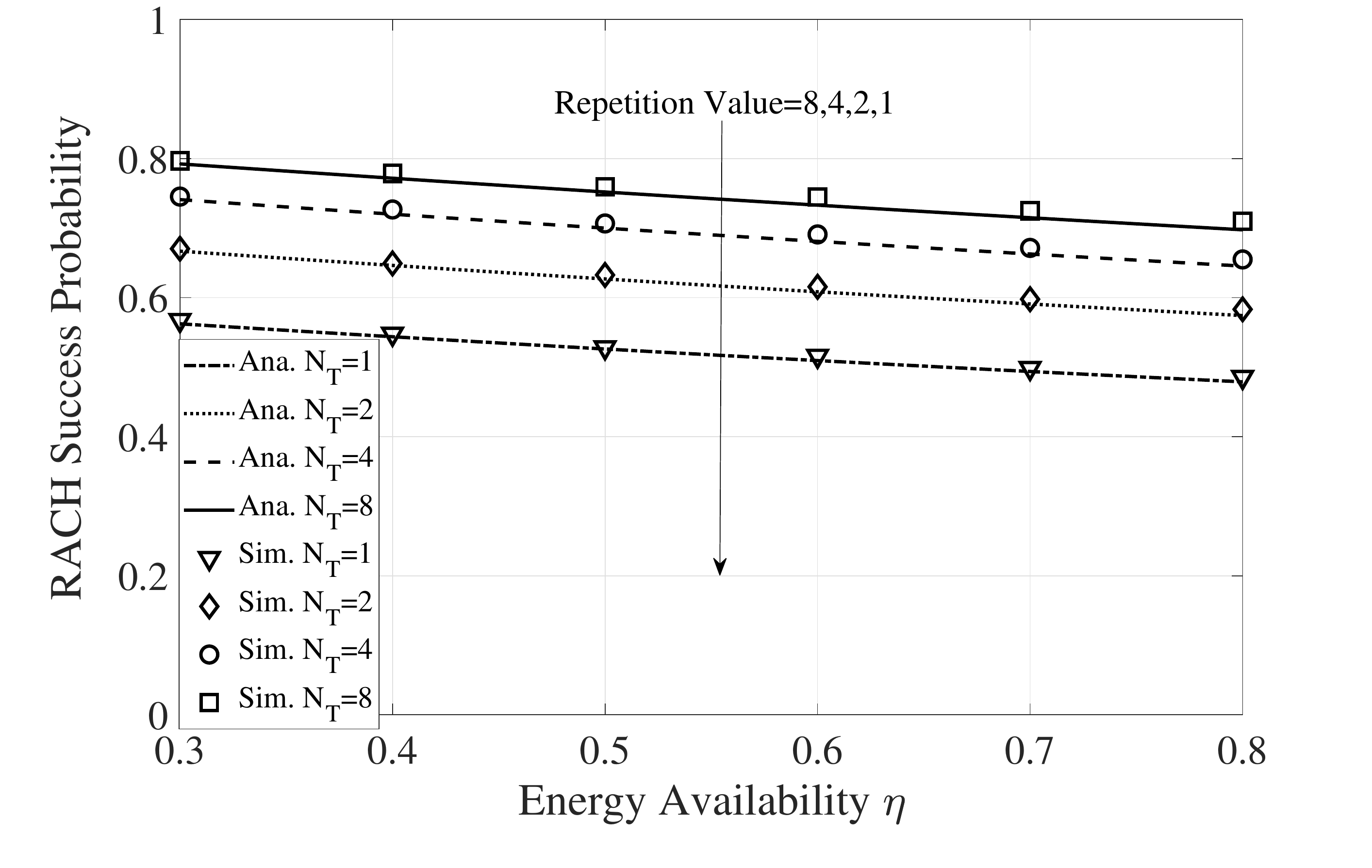}
			%	\vspace*{-0.4cm}
				\caption{\scriptsize {RACH success probability versus the energy availability  $\eta_0$ for various repetition values $N_T$.}}
			\end{minipage}
			\label{fig:7}
%			\begin{minipage}[t]{0.48\textwidth}
%				\centering
%				\includegraphics[width=3in,height=2.4in]{0624_power.eps}
%				\vspace*{-0.4cm}
%				\caption{\scriptsize {RACH success probability versus the transmission powers $P$ for various repetition values $N_T$.}}
%			\end{minipage}
%			\label{fig:5}
		\end{center}
	\end{figure}
	
	Fig. 8 plots the RACH success probability of a randomly chosen IoT device versus the transmission power $P$ for various preamble repetition values $N_T$.
	We first observe a good match between the analysis and the simulation results, which validates the accuracy of the developed mathematical framework.
	As expected, we observe that the RACH success probability increases as the transmission power of IoT devices increase.
	This can be explained by the reason that whilst increasing the transmission power $P$ leads to higher interference power, it also leads to increased received signal power, thereby improves the overall SIR and hence the RACH success probability. 
	It is worth noting that the RACH success probability increases faster at first (e.g., when $P\le 0.1$) and then gradually becomes steady, which reveals that there is a limit value of the cell maximum transmission power.
	Interestingly, we observe that the RACH success probabilities with a higher repetition value, e.g., $N_T=8$ become stable earlier than those with a lower repetition value, e.g., $N_T=1$, due to that the higher chance of the RACH to succeeds in higher repetition value case.
	
	\begin{figure}[htbp!]
		\begin{center}
%			\begin{minipage}[t]{0.48\textwidth}
%				\centering
%				\includegraphics[width=3.0in,height=2.4in]{0627_energy1.eps}
%				\vspace*{-0.4cm}
%				\caption{\scriptsize {RACH success probability versus the energy availability  $\eta_0$ for various repetition values $N_T$.}}
%			\end{minipage}
		%	\label{fig:7}
						\begin{minipage}[t]{0.48\textwidth}
							\centering
							\includegraphics[width=3in,height=2.4in]{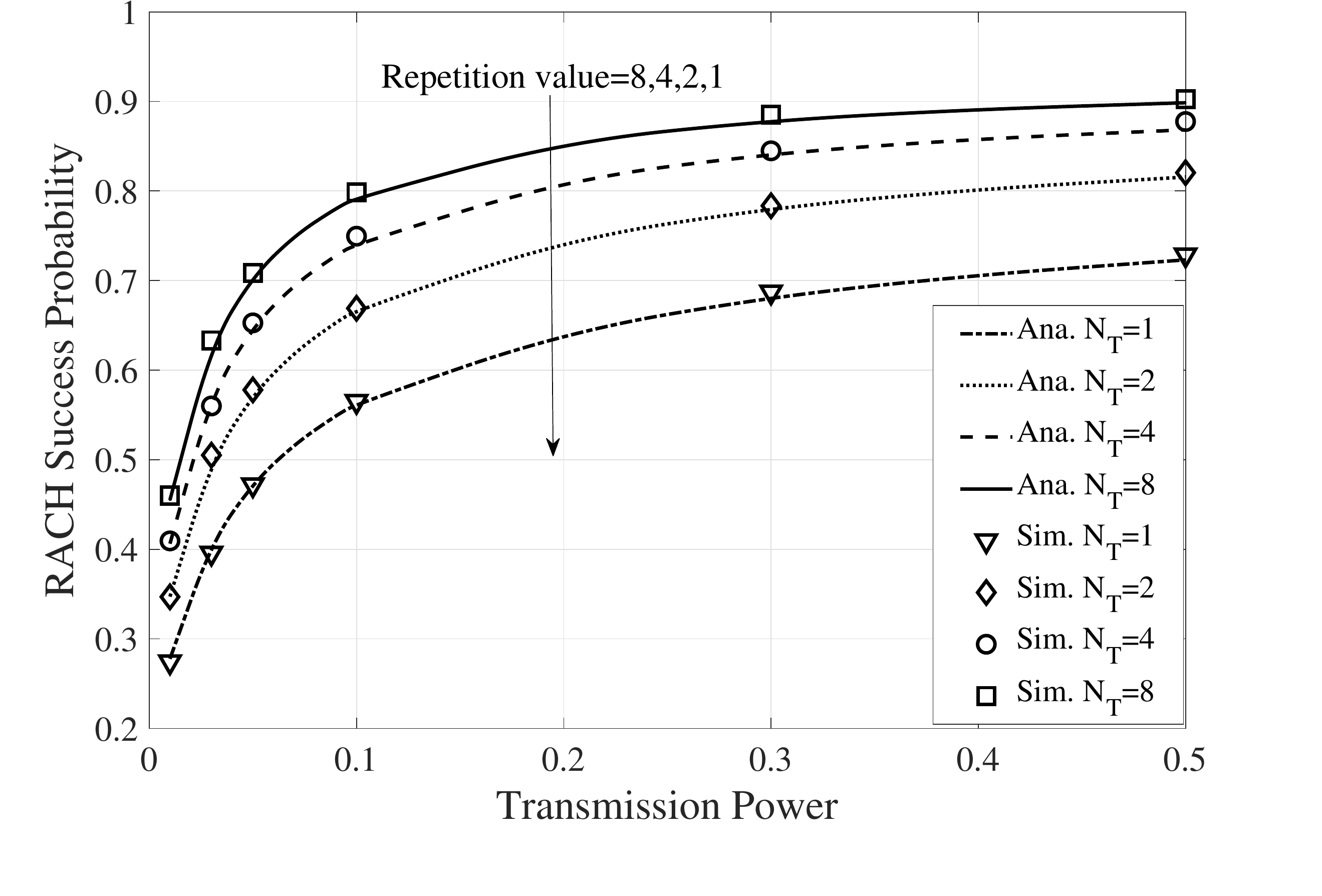}
						%	\vspace*{-0.4cm}
							\caption{\scriptsize {RACH success probability versus the transmission powers $P$ for various repetition values $N_T$.}}
						\end{minipage}
						\label{fig:8}
		\end{center}
	\end{figure}

	Fig. 9 plots the RACH success probability of a randomly chosen IoT device versus the SINR threshold $\gamma_{th}$  for various preamble repetition values $N_T$.
	As expected, the RACH success probability degrades with an increase in the SINR threshold. 
	According to Fig. 4, increasing SINR threshold $\gamma_{th}$ leads to lower preamble transmission success probability but higher non-collision probability, thereby decreases the overall RACH success probability. 
	There is a tradeoff between preamble transmission success probability and non-collision probability.
	
	% \begin{figure}[htbp!]
	%  	\centering
	%  	\includegraphics[width=4in,height=2.6in]{0624_SIN_1.eps}
	%  	\caption{{RACH success probability versus the SINR threshold $\gamma_{th}$ for various repetition values $N_T$.
	%  	}}
	%  	\label{fig:9}
	%  \end{figure}
	
	\begin{figure}[htbp!]
		\begin{center}
			\begin{minipage}[t]{0.48\textwidth}
				\centering
				\includegraphics[width=3in,height=2.4in]{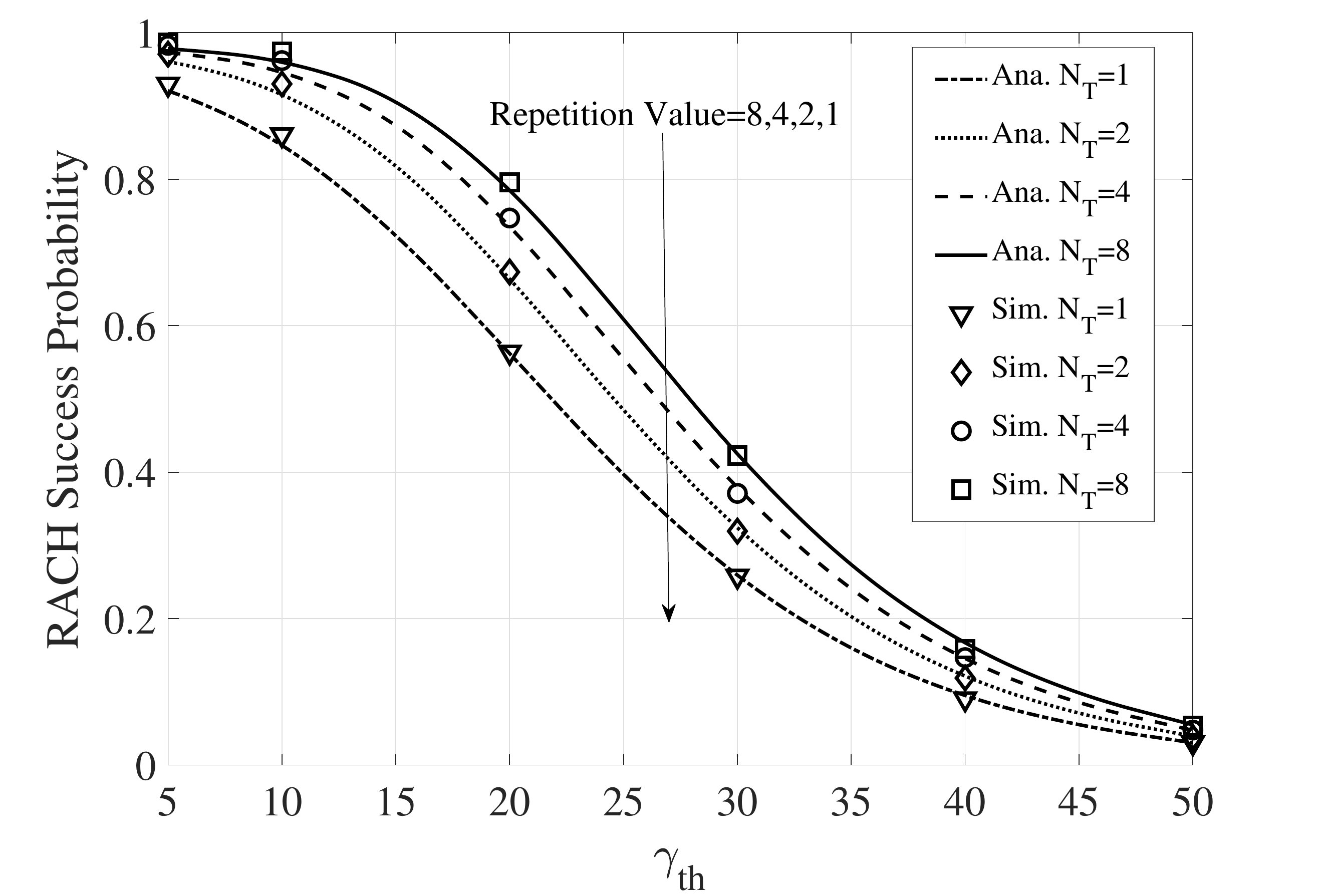}
			%	\vspace*{-0.4cm}
				\caption{\scriptsize RACH success probability versus the SINR threshold $\gamma_{th}$ for various repetition values $N_T$.}
			\end{minipage}
			\label{fig:9}
%			\begin{minipage}[t]{0.48\textwidth}
%				\centering
%				\includegraphics[width=3.0in,height=2.4in]{1031_lambda_11_R.eps}
%				\vspace*{-0.4cm}
%				\caption{\scriptsize {RACH success probability versus  density ratio $\lambda_D/\lambda_B$ for various repetition values $N_T$.}}
%			\end{minipage}
%			\label{fig:11}
		\end{center}
	\end{figure}

	Fig. 10(a) and Fig. 10(b) plot the RACH success probability of a randomly chosen IoT device versus the density ratio $\lambda_D/\lambda_B$  in the network with and without inter-cell interference, respectively.
	%Fig. 12 plots the RACH success probability for networks with perfect PRACH root sequence planning (e.g., no inter-cell interference).
	We first observe that the RACH success probability decreases with the increase of the density ratio between IoT devices and BSs ($\lambda_D/\lambda_B$), which is due to the following two reasons: 1) increasing the number of IoT devices generating interference leads to lower received SINR at the eNB; 2) increasing the number of IoT devices leads to higher probability of collision. 
	In addition, increasing repetition value increases the RACH success probability due to that it offers more opportunities to re-transmit a preamble with the time and frequency diversity.

		\begin{figure}[htbp!]
		\begin{center}
%			\begin{minipage}[t]{0.48\textwidth}
%				\centering
%				\includegraphics[width=3in,height=2.4in]{0624_SIN_1.eps}
%				\vspace*{-0.4cm}
%				\caption{\scriptsize RACH success probability versus the SINR threshold $\gamma_{th}$ for various repetition values $N_T$.}
%			\end{minipage}
%			\label{fig:9}
						\begin{minipage}[t]{0.48\textwidth}
							\centering
							\includegraphics[width=3.2in,height=2.4in]{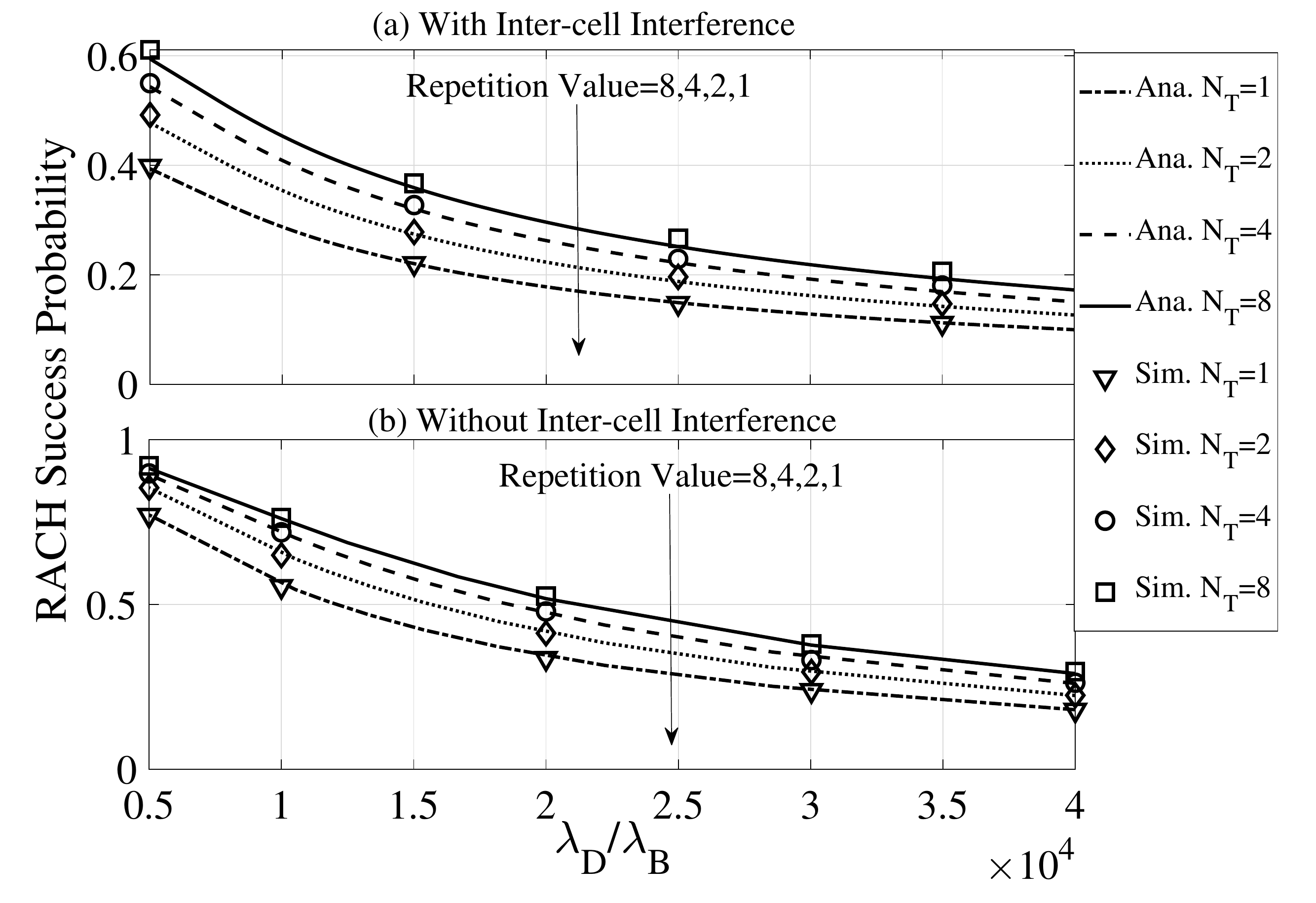}
						%	\vspace*{-0.4cm}
							\caption{\scriptsize {RACH success probability versus  density ratio $\lambda_D/\lambda_B$ in the network with and without inter-cell interference.}}
						\end{minipage}
						\label{fig:10}
		\end{center}
	\end{figure}

	\begin{figure}[htbp!]
		\centering
		\includegraphics[width=3.2in,height=2.6in]{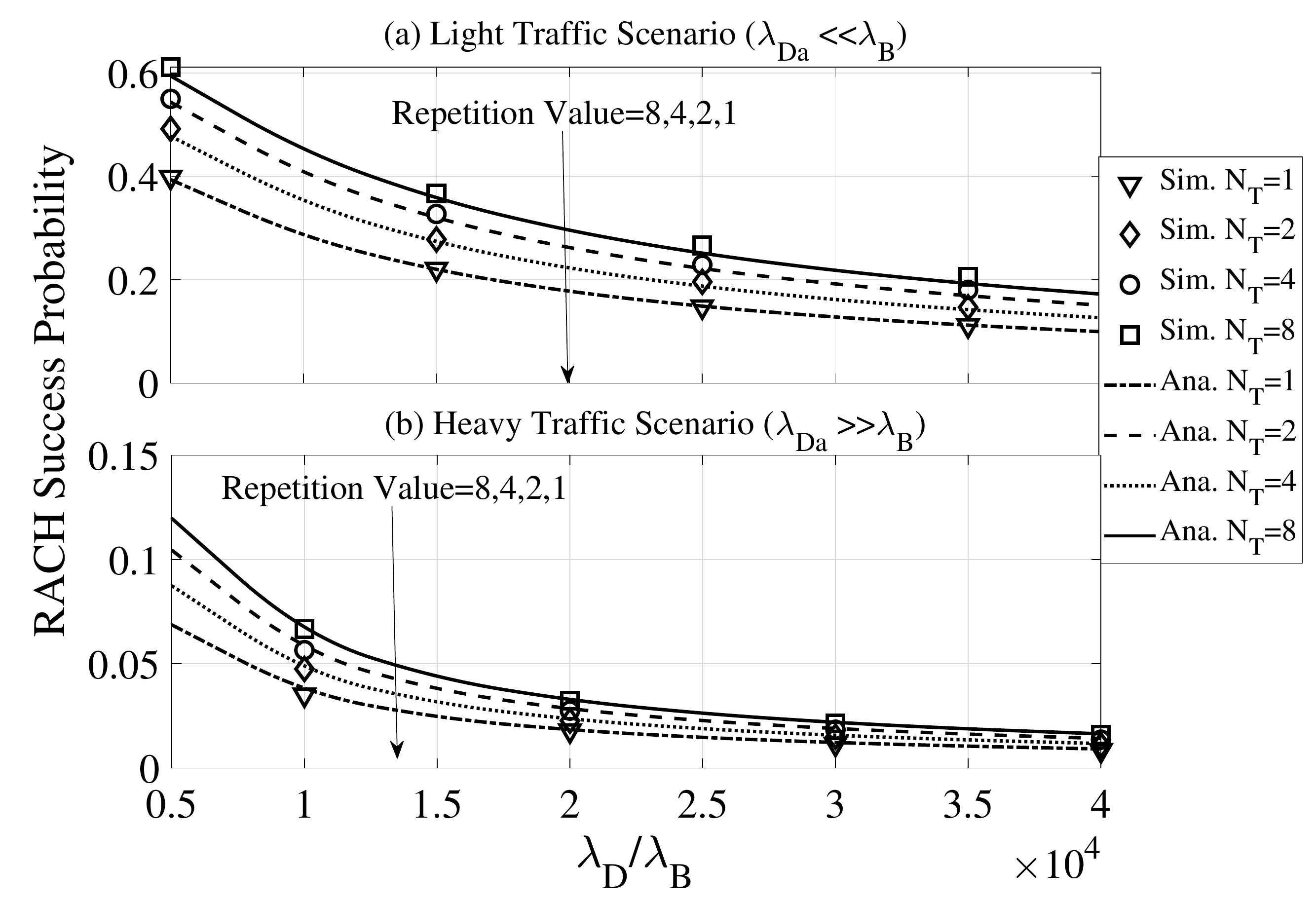}
		\caption{{RACH success probability versus  density ratio $\lambda_D/\lambda_B$ in the network with light traffic and heavy traffic.
		}}
		\label{fig:11}
	\end{figure}
	
Fig. 11(a) and Fig. 11(b) plot the RACH success probability of a randomly chosen IoT device versus the density ratio $\lambda_D/\lambda_B$  in the network with light traffic ($\mathcal{A}_a=0.001$, $\lambda_{Da}\ll\lambda_B$, and $\varepsilon=1$) and heavy traffic ($\mathcal{A}_a=0.015$, $\lambda_{Da}\gg\lambda_B$, and $\varepsilon=1.25$), respectively.
		In the light traffic scenario, RACH success probability increase when increasing the repetition value. However, in the heavy traffic scenario, the RACH success probability cannot improve much when the repetition value increased.
		That is to say, the repetition scheme can efficiently improve the RACH success probability in a light traffic scenario, but only slightly improves that performance with very inefficient channel resource utilization in a heavy traffic scenario.

	\begin{figure}[htbp!]
		\begin{center}
			\begin{minipage}[t]{0.48\textwidth}
				\centering
				\includegraphics[width=3.2in,height=2.6in]{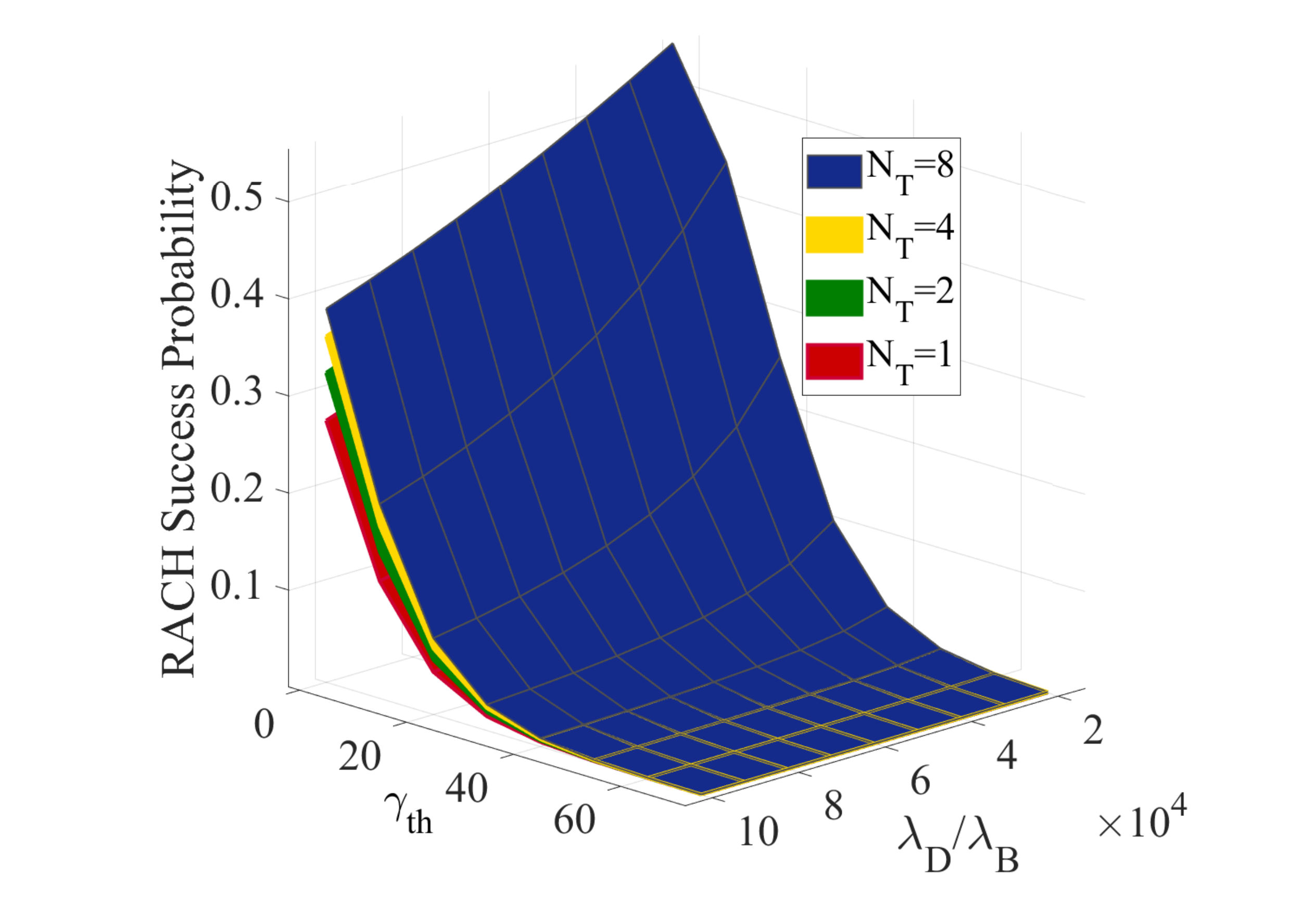}
				\vspace*{-0.4cm}
				\caption{\scriptsize {RACH success probability of a randomly chosen IoT device versus the densities ratio $\lambda_D/\lambda_B$ and SINR threshold $\gamma_{th}$.}}
			\end{minipage}
			\label{fig:12}
			\begin{minipage}[t]{0.48\textwidth}
				\centering
				\includegraphics[width=3.5in,height=2.8in]{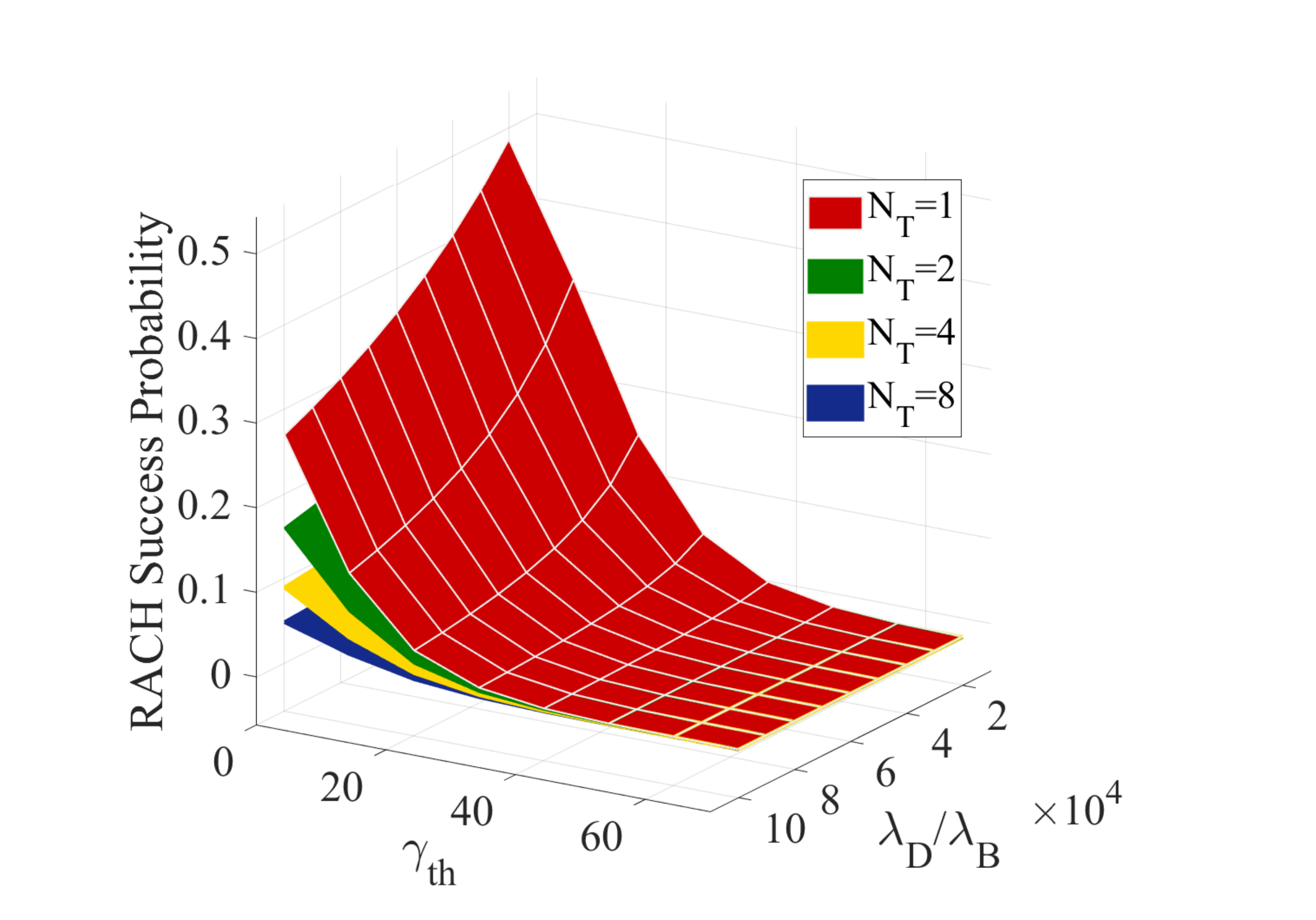}
			%	\vspace*{-0.4cm}
				\caption{\scriptsize {Repetition efficiency of a randomly chosen IoT device versus the densities ratio $\lambda_D/\lambda_B$ and SINR threshold $\gamma_{th}$.}}
			\end{minipage}
			\label{fig:13}
		\end{center}
	\end{figure}

	Fig. 12 and Fig. 13 plot the RACH success probability and the Repetition  efficiency (i.e., $\zeta = {{\cal P}_0}/{N_T}$)  of a randomly chosen IoT device  versus the density ratio $\lambda_D/\lambda_B$ and SINR thresholds $\gamma_{th}$, respectively.
	Obviously, increasing the repetition value increases the success opportunities of RACH, but decreases the repetition efficiency, which reveals that if the repetition value is overestimated, the IoT device will waste the potential resource for data transmission and lead to lower resource efficiency.
	In addition, we observe that the repetition efficiencies decrease seriously when increasing the repetition value from $N_T=1$ to  $N_T=2$, especially for the lower density ratio.
	% repetition values become smaller when the repetition values increase. That is to say, the improvement degree of the repetition value to the RACH access probability decreases when increasing the repetition value, escepilly fo the 
	This is due to the fact that when there are fewer active IoT devices (e.g., $\lambda_D/\lambda_B=10^4$) to contention for the same resources, it is  likely to succeed with small repetition values,  which contributes to a relatively high channel resource utilization.

	%Fig. 13 plots the Repetition efficiency $\zeta = {{\cal P}}/{N_T}$ of the 1st IoT device versus the energy availability for various preamble repetition values. 
	%Obviously, increasing the repetition value increases the success opportunities of RACH, but decreases the repetition efficiency.
	%In addition, the gaps between the repetition efficiencies of different repetition values become smaller when the energy availability increases.
	%This is due to the fact that when there are fewer active IoT devices (e.g., $\eta_i=0.3$) to contention for the same resources, it is more likely to succeed with fewer repetition values, vice versa.
	%The results give insights that when there are fewer (more) active IoT devices, it is better to choose smaller (larger) repetition values. 

	%As shown in Fig. 4, the curves representing the analytical results closely match with simulation points that validates the accuracy of developed mathematical framework. 

	\section{Conclusion}
	In this paper, we have presented a comprehensive RACH success probability analysis for the NB-IoT network with energy harvesting from natural resources.
	% Using stochastic geometry, we characterized the distribution of distances from the serving device and interfering devices to eNB.
	% Using tools from Markov stochastic process, we quantified the uncertainty in IoT device energy availability due to the finite battery capacity and inherent randomness in energy harvesting. 
	We have derived a closed-form expression of the energy availability and analyzed the results for some realistic strategies. 
	We found that energy availability remains unchanged if the network is properly designed despite randomness in energy harvesting.
	We further analyzed the RACH under the repetition transmission scheme in the NB-IoT system based on energy availability.
	We derived the exact expression for the RACH success probability under time-correlated interference. 
	Different from existing works, we considered both SINR outage and collision from the network point of view in the NB-IoT network, where each IoT device adopts fixed transmission power.
Our results have shown that increasing the repetition value increases the RACH success probability but decreases the repetition efficiency. 
		In the light traffic scenario, RACH success probability can meet the requirement with small repetition values. However, in the heavy traffic scenario,  very high repetition value leads to a low channel resources utilization.
	Our results  also have shown that there is an upper limit on transmission power and too large transmission power will waste energy.

	\appendices
	\numberwithin{equation}{section}
	\section{a proof of lemma 1 }
	Let $A_0=-B_0$ and then we have 
	\begin{align}\label{I}
	A_0{A_0}^{-1}=I,
	\end{align}
	where ${A_0}^{-1}$ is the inverse of the matrix ${A_0}$. Let 
	\begin{align}\label{inverse}
	{A_0}^{-1}=(a_{mn})_{1\le m,n\le M_0}=(A_1,A_2,\cdots,A_{M_0}),
	\end{align}
	and
	\begin{align}\label{unit}
	I=(e_{mn})_{1\le m,n\le M_0}=(E_1,E_2,\cdots,E_{M_0}),
	\end{align}
	where $A_r$ and $E_r$ denote the $r$th column of ${A_0}^{-1}$ and $I$ respectively, $r=1, 2, \cdots,M_0$.
	Then (A.1) can be rewritten in the form
	\begin{align}\label{M_i}
	A_0A_{r}=E_{r}.
	\end{align}
	
	In order to calculate ${A_0}^{-1}$, we begin by considering the elements of the first column $A_1=(a_{11}, a_{21},\cdots,a_{M_01})^T$ of ${A_0}^{-1}$.
	That means
	\begin{align}\label{E_i}
	&
	\begin{bmatrix}
	\mu _0 + \nu _0 &  - \mu _0 & 0 &  \cdots & 0 & 0 \\
	- \nu _0 & \mu _0+\nu _0  & - \mu _0 & \cdots  & 0 & 0\\
	\vdots & \vdots & \vdots  &          \ddots &  \vdots               &  \vdots \\
	0 & 0 &  0     &    \cdots  &   {\mu _0} + {\nu _0} & -\mu_0\\
	0 & 0 &  0  & \cdots & -\nu _0 &\nu _0
	\end{bmatrix}
	\begin{bmatrix}
	{a_{1{1}}}\\
	{a_{2{1}}}\\
	\vdots \\
	{a_{{M_{0 - 1}}{1}}}\\
	{a_{{M_0}{1}}}
	\end{bmatrix}
	\nonumber\\
&	=
	\begin{bmatrix}
	1 & 0 & \cdots & 0 & 0
	\end{bmatrix}
	^T.
	\end{align}
	Then we have $A_1$=$(1/\nu_0, 1/\nu_0,\cdots,1/\nu_0)^T$. Plugging $r=2, 3,\cdots,M_0$ into (A.4) respectively, we have all the elements of ${A_0}^{-1}$, i.e, $(-{B_0})^{-1}$.
	
	% {
	% \section{a proof of lemma 2 }
	% According to \cite{1512427}, the probability of finding $k$ nodes in
	% a bounded Borel $A \subset {{\mathbb R}_m}$ in a homogeneous $m$-dimensional Poisson point process of intensity $\lambda$ is given by
	% \begin{align}
	% {\rm{P_r}}({K = k} ) = {e^{ - \lambda \mu (A)}}\frac{{{{( {\lambda \mu (A)} )}^k}}}{{k!}},
	% \end{align}
	% where $K$ is the Poisson random variable, and $\mu (A)$ is the
	% standard Lebesgue measure of A.
	% }
	
	% {
	% Thus, the probability of finding zero nodes in a bounded Borel $A \subset {{\mathbb R}_2}$ in a homogeneous 2D Poisson point process of
	% intensity $\lambda_B$ is obtained as
	% \begin{align}
	% {\rm{P_r}}({K = 0} ) = {e^{ - {\lambda_B} \mu (A)}},
	% \end{align}
	% where $\mu (A)=\pi r_0^2$.
	% }
	
	% {
	% Using $f_{R_0}(r_0)= -\displaystyle\frac{{d{\rm{P_r}}(K = 0)}}{{d{r_0}}}$, we prove (\ref{r_0}).
	% }

	\section{a proof of lemma 2 }
	{We note that the preamble transmission success probability in \eqref{SI} depends on the transmission distance $r_0$.
		According to the PDF of $r_0$ given in \eqref{r_0}, we have
	}
	{
		\begin{align}\label{conditionr_0}
		&p_i({\gamma _{th}}) = {\mathbb E_{{R_0}}}\Big[{\mathbb P}_0[ {{\theta _1}(r_0),{\theta _2}(r_0),...,{\theta _{{n_T}}}(r_0)| {{r_0}} } ]\Big]\nonumber\\
		%& = \int\limits_{{D_{i - 1}}}^{{D_i}} {{\rm P}_k^i\left[ {{\theta _1}(k),{\theta _1}(k),...,{\theta _{{m_i}}}(k)\left| {{r_i}} \right.} \right]} {f_{{R_i}}}({r_i})d{r_i}\nonumber\\
		&= \int{{\mathbb P}_0[ {{\theta _1}(r_0),{\theta _2}(r_0),...,{\theta _{{n_T}}}(r_0)| {{r_0}}} ]} f(r_0)d{r_0}\nonumber \\
		&= \int{{\mathbb P}_0}[{\rm SINR_1}({r_0}) \geqslant {\gamma _{th}}, \cdots ,{{\rm SINR}_l}({r_0}) \geqslant {\gamma _{th}}\big| {{r_0}}\big]f(r_0)d{r_0} \nonumber \\
		&= \int{{\mathbb P}_0}\Big[{h_{0}^1} \geqslant \frac{{{\gamma _{th}}r_0^\alpha ({\mathcal {I}_{0}^1} + {\sigma ^2})}}{P},\nonumber \\&\qquad \qquad ...,{h_{0}^l} \geqslant \frac{{{\gamma _{th}}r_0^\alpha ({{\mathcal {I}_{0}^l}} + {\sigma ^2})}}{P}\Big| {{r_0}} \big.\Big]f(r_0)d{r_0} \nonumber \\
		&\overset{\text{(a)}}= \int{\mathbb{E}}\Big[ \exp \Big( - {\frac{{{\gamma _{th}}r_0^\alpha }}{P}}({{\mathcal {I}_{0}^1}} + {\sigma ^2})\Big) \nonumber \\ 
		&\qquad \qquad...\exp \Big( - {{\frac{{{\gamma _{th}}r_0^\alpha }}{P}}}({{\mathcal {I}_{0}^l}} + {\sigma ^2})\Big)\Big| {{r_0}} 
		\big. \Big]f(r_0)d{r_0} \nonumber \\
		% &= {\exp ( - l{v_i}{\sigma ^2}){\mathbb{E}}\big[ {\exp ( - {v_i}{I_{i,1}})\exp ( - {v_i}{I_{i,2}})...\exp ( - {v_i}{I_{i,l}})\big| {{r_1,r_2,...,r_{n+1}},N = n} \big.} \big]} \nonumber \\
		&= \int\exp \Big(  {{\frac{{{-l\gamma _{th}\sigma ^2}r_0^\alpha }}{P}}}{}\Big)\nonumber \\
		&\times{\mathbb{E}}\Big[ {\exp \Big(  {{\frac{{-{\gamma _{th}}r_0^\alpha }}{P}}}\sum\limits_{\beta  = 1}^l {{{\mathcal {I}_{0}^\beta}}\Big)\Big| {{r_0}} \big.} } \Big]f(r_0)d{r_0},
		%&=\int_0^\infty  {} 2\pi{\lambda_B}r_0\exp(-{\lambda_B}{\pi}r_0^2{{-\frac{{{l\gamma _{th}\sigma ^2}r_0^\alpha }}{P}}})\exp \Big( \overset{\text{}}  { - 2\pi {\lambda _{Da}}\int_{0}^\infty  {\Big[ {1 - \Big({}{{1 + {\gamma _{th}}r_0^\alpha{y^{ - \alpha }}}}\Big)^{-l}} \Big]} ydy} \Big)d{r_0},
		\end{align}
		where (a) follows from the independence of $h_{0}^l$, and $\mathcal I_{0}$ is given in \eqref{INTRA}.
	}

	The Laplace transform of the aggregate interference is characterized according to the definition of ${{\mathcal L_{I_0}}(s) = {{\mathbb E}_{I_0}}({e^{ - sI_0}})}$ as 
	\begin{align}\label{laplace}
	&{\mathbb{E}}\Big[ {\exp \Big( - {{\frac{{{\gamma _{th}}r_0^\alpha }}{P}}}\sum\limits_{\beta  = 1}^l {{{\mathcal {I}_{0}^\beta}}\Big)\Big| {{r_0}} \big.} } \Big] 
	% &= {\mathbb E}\bigg[ {\exp \Big( - {v_i}\sum\limits_{j = 1,j \ne i}^{n + 1} {P\sum\limits_{\beta  = 1}^l {{h_{j,\beta }}} r_j^{ - \alpha }\Big)\Big| {{r_1,r_2,...,r_{n+1}},N = n} \big.} } \bigg] 
	\nonumber \\
	&= {\mathbb E}\bigg[ {\prod\limits_{j \in  \mathcal{Z_D}}^{} {\exp \Big( - {\gamma _{th}}r_0^\alpha \sum\limits_{\beta  = 1}^l {{h_{j}^\beta }} r_j^{ - \alpha }\Big)\Big| {{r_0}} \big.} } \bigg] \nonumber \\
	&\overset{\text{(a)}}= {\mathbb E}\Big[ {\prod\limits_{j \in  \mathcal{Z_D}}^{} \Big({\frac{1}{{1 + {\gamma _{th}}r_0^\alpha r_j^{ - \alpha }}}}\Big)^l } \Big]  \\ \nonumber
	& \overset{\text{(b)}}= \exp \Big( { - 2\pi {\lambda _{Da}}\int_{0}^\infty  {\Big[ {1 - \Big({}{{1 + {\gamma _{th}}r_0^\alpha{y^{ - \alpha }}}}\Big)^{-l}} \Big]} ydy} \Big) 
	% &= {{\mathbb{E}}_{{R_j}}}\bigg[ {\prod\limits_{j = 1,j \ne i}^{n + 1} {{{\Big( {\frac{1}{{1 + {\gamma _{th}}r_i^\alpha r_j^{ - \alpha }}}} \Big)}^l}\Big| {{r_i},{r_j},N = n} \big.} } \bigg]{f_{{R_j}}}({r_j})d{r_j} 
	\end{align}
	where (a) is obtained by taking the average with respect to $h_{j}^\beta$ and (b) follows from the probability generation functional (PGFL) of the PPP. 
	Substituting (\ref{laplace}) into \eqref{conditionr_0}, we verified (\ref{conditionr}) in Lemma 2.

	\ifCLASSOPTIONcaptionsoff
	\newpage
	\fi

	\bibliographystyle{IEEEtran}

	\bibliography{IEEEabrv,energy}
	
\end{document}